\documentclass{article}
\usepackage{fullpage}

\usepackage[utf8]{inputenc}
\usepackage{amsmath}
\usepackage{amsthm,amssymb}
\usepackage{nicefrac,amsfonts}
\usepackage{thmtools,mathtools,leftindex,leftidx}
\usepackage{thm-restate}
\usepackage{caption,subcaption}
\usepackage{epsfig,graphicx,graphics,color,tikz,tikz-cd,tcolorbox}
\usepackage{enumerate,enumitem}
\usepackage[sort,nocompress]{cite}
\usepackage{xspace}
\usepackage{bbm}
\usepackage{hyperref}
\usepackage[capitalise]{cleveref}
\usepackage{float}

\hypersetup{
    colorlinks=true,       % false: boxed links; true: colored links
    linkcolor=blue,        % color of internal links (change box color with linkbordercolor)
    citecolor=magenta,         % color of links to bibliography
    filecolor=magenta,     % color of file links
    urlcolor=cyan,         % color of external links
    linktoc=all
}
\usepackage{titling}
\usepackage{tikz}
\usetikzlibrary{calc, arrows.meta}

\newlength{\bibitemsep}
\newlength{\bibparskip}
\let\oldthebibliography\thebibliography
\renewcommand\thebibliography[1]{%
  \oldthebibliography{#1}%
  \setlength{\parskip}{\bibitemsep}%
  \setlength{\itemsep}{\bibparskip}%
}

\makeatletter
\renewcommand{\paragraph}{%
  \@startsection{paragraph}{4}%
  {\z@}{1.6ex \@plus 1ex \@minus .2ex}{-0.5em}%
  {\normalfont\normalsize\bfseries}%
}
\makeatother

\DeclareFontFamily{U}{mathx}{\hyphenchar\font45}
\DeclareFontShape{U}{mathx}{m}{n}{
  <5> <6> <7> <8> <9>
  <10> <10.95> <12> <14.4> <17.28> <20.74> <24.88>
  mathx10
}{}
\DeclareSymbolFont{mathx}{U}{mathx}{m}{n}
\DeclareMathAccent{\widebar}{0}{mathx}{"73}

\theoremstyle{plain}
\newtheorem{thm}{Theorem}[section]
\newtheorem{lem}[thm]{Lemma}
\newtheorem{cor}[thm]{Corollary}
\newtheorem{cla}[thm]{Claim}
\newtheorem{prop}[thm]{Proposition}

\theoremstyle{definition}

\newtheorem{rem}[thm]{Remark}

\newcommand*{\claimproofname}{Proof.}
\newenvironment{claimproof}[1][\claimproofname]{\begin{proof}[#1]}{\end{proof}}

\makeatletter
\newcommand{\leqnomode}{\tagsleft@true}
\newcommand{\reqnomode}{\tagsleft@false}
\makeatother

\makeatletter
 \newcommand{\linkdest}[1]{\Hy@raisedlink{\hypertarget{#1}{}}}
\makeatother

\def\final{0}  % set this to 1 to get a comment-free version
\def\iflong{\iffalse}
\ifnum\final=0  %namely if we allow comments in the output
\newcommand{\kristof}[1]{{\color{red}[{ \tiny \textbf{Kristóf:}  #1}]\marginpar{\color{red}*}}}
\newcommand{\tamas}[1]{{\color{blue}[{ \tiny \textbf{Tamás:}  #1}]\marginpar{\color{blue}*}}}
\newcommand{\andris}[1]{{\color{magenta}[{ \tiny \textbf{Andris:}  #1}]\marginpar{\color{magenta}*}}}
\newcommand{\florian}[1]{{\color{teal}[{ \tiny \textbf{Florian:}  #1}]\marginpar{\color{teal}*}}}
\else % in this case [final=1] we don't want any comments to show
\newcommand{\kristof}[1]{}
\newcommand{\tamas}[1]{}
\newcommand{\andris}[1]{}
\newcommand{\florian}[1]{}
\fi

\definecolor{myred}{HTML}{A71D31}
\definecolor{myblue}{HTML}{355691}
\definecolor{mygreen}{HTML}{4DA167}

\newcommand{\bZ}{\mathbb{Z}}

\newcommand{\cC}{\mathcal{C}}

\newcommand{\cF}{\mathcal{F}}
\newcommand{\cH}{\mathcal{H}}

\newcommand{\cP}{\mathcal{P}}

\newcommand{\cS}{\mathcal{S}}
\newcommand{\cU}{\mathcal{U}}

\newtcolorbox{probbox}{arc=6pt,
                      colback=white!100,
                      colframe=black!50,
                      before skip=6pt,
                      after skip=6pt,
                      boxsep=1pt,
                      left=6pt,
                      right=6pt,
                      top=4pt,
                      bottom=4pt}

\usepackage{accents}
\newcommand{\dGaux}[2]{%
  \vphantom{#2}%
  \overset{%
    \smash{%
      \raisebox{%
        \ifx#1\scriptstyle -0.4ex % subscript
        \else\ifx#1\scriptscriptstyle -0.2ex % nested subscript
        \else -0.6ex % main text
        \fi\fi
      }{$\scriptstyle\rightharpoonup$}%
    }%
  }{#2}%
}
\newcommand{\dG}{\mathpalette\dGaux G}
\newcommand{\dP}{\mathpalette\dGaux P}
\newcommand{\dGp}{\mathpalette\dGaux {G'}}
\newcommand{\dGpo}{\mathpalette\dGaux {G'_1}}
\newcommand{\dGpt}{\mathpalette\dGaux {G'_2}}
\newcommand{\dH}{\mathpalette\dGaux H}
\newcommand{\dHp}{\mathpalette\dGaux {H'}}
\newcommand{\dA}{\mathpalette\dGaux A}
\newcommand{\dE}{\mathpalette\dGaux E}
\newcommand{\de}{\mathpalette\dGaux e}
\newcommand{\df}{\mathpalette\dGaux f}
\newcommand{\dg}{\mathpalette\dGaux g}

\newcommand{\ktori}{\textsc{Steiner Rooted $k$-Orientation with $t$ Terminals}\xspace}

\newcommand{\ltori}{\textsc{Steiner Rooted $\ell$-Orientation with $t$ Terminals}\xspace}
\newcommand{\kori}{\textsc{Steiner Rooted $k$-Orientation}\xspace}
\newcommand{\twoori}{\textsc{Steiner Rooted $2$-Orientation}\xspace}
\newcommand{\mtori}{\textsc{Modified Steiner Rooted Orientation with $t$ Terminals}\xspace}
\newcommand{\mthreeori}{\textsc{Modified Steiner Rooted Orientation with $3$ Terminals}\xspace}
\newcommand{\tori}{\textsc{Steiner Rooted Orientation with $t$ Terminals}\xspace}
\newcommand{\threeori}{\textsc{Steiner Rooted Orientation with $3$ Terminals}\xspace}
\newcommand{\fourori}{\textsc{Steiner Rooted Orientation with $4$ Terminals}\xspace}
\newcommand{\ori}{\textsc{Steiner Rooted Orientation}\xspace}
\newcommand{\roria}{\textsc{$R$-Orientation with $\alpha$-Demand}\xspace}
\newcommand{\rori}{\textsc{$R$-Orientation}\xspace}

\newcommand{\true}{\texttt{True}}
\newcommand{\false}{\texttt{False}}

\title{Fixed-Parameter Tractability and Hardness for Steiner Rooted and Locally Connected Orientations}

\thanksmarkseries{alph}
\author{
Kristóf Bérczi\thanks{MTA-ELTE Matroid Optimization Research Group and HUN-REN-ELTE Egerváry Research Group, Department of Operations Research, ELTE Eötvös Loránd University, and HUN-REN Alfréd Rényi Institute of Mathematics, Budapest, Hungary. Email: \texttt{kristof.berczi@ttk.elte.hu}.}
\and
Florian Hörsch\thanks{CISPA Helmholtz Center for Information Security, Saarbrücken, Germany. Email: \texttt{florian.hoersch@cispa.de}.}
\and
András Imolay\thanks{Department of Operations Research, ELTE Eötvös Loránd University, Budapest, Hungary. Email: \texttt{andras.imolay@ttk.elte.hu}.}
\and
Tamás Schwarcz\thanks{Department of Mathematics, London School of Economics and Political Science, London, England, United Kingdom. Email: \texttt{t.b.schwarcz@lse.ac.uk}.} 
}
\date{}

\begin{document}
\maketitle

\thispagestyle{empty}
%%%%%%%%%%%%%%%%%%%%%%%%%%%%%%%%
\begin{abstract} 
Finding a Steiner strongly $k$-arc-connected orientation is particularly relevant in network design and reliability, as it guarantees robust communication between a designated set of critical nodes. Király and Lau (FOCS 2006) introduced a rooted variant, called the Steiner Rooted Orientation problem, where one is given an undirected graph on $n$ vertices, a root vertex, and a set of $t$ terminals. The goal is to find an orientation of the graph such that the resulting directed graph is Steiner rooted $k$-arc-connected. This problem generalizes several classical connectivity results in graph theory, such as those on edge-disjoint paths and spanning-tree packings. While the maximum $k$ for which a Steiner strongly $k$-arc-connected orientation exists can be determined in polynomial time via Nash-Williams' orientation theorem, its rooted counterpart is significantly harder: the problem is NP-hard when both $k$ and $t$ are part of the input. In this work, we provide a complete understanding of the problem with respect to these two parameters. In particular, we give an algorithm that solves the problem in time $f(k,t)\cdot n^{O(1)}$, establishing fixed-parameter tractability with respect to the number of terminals $t$ and the target connectivity $k$. We further show that the problem remains NP-hard if either $k$ or $t$ is treated as part of the input, meaning that our algorithm is essentially optimal from a parameterized perspective. Importantly, our results extend far beyond the Steiner setting: the same framework applies to the more general orientation problem with local connectivity requirements, establishing fixed-parameter tractability when parameterized by the total demand and thereby covering a wide range of arc-connectivity orientation problems.
\medskip


\end{abstract}
%%%%%%%%%%%%%%%%%%%%%%%%%%%%%%%%
\newpage
\pagenumbering{roman}
\tableofcontents
\newpage
\pagenumbering{arabic}
\setcounter{page}{1}
%%%%%%%%%%%%%%%%%%%%%%%%%%%%

%%%%%%%%%%%%%%%%%%%%%%%%%%%%%%%%
\section{Introduction}
\label{sec:intro}
%%%%%%%%%%%%%%%%%%%%%%%%%%%%%%%%

A {\it Steiner strongly $k$-arc-connected orientation} of a graph is an orientation of an undirected graph $G=(V,E)$ with a set of terminals $S \subseteq V$ such that, for every pair of terminals $u,v \in S$, there are at least $k$ pairwise arc-disjoint directed paths from $u$ to $v$. Steiner strongly $k$-arc-connected orientations are of interest in network design and reliability, where one aims to ensure robust communication or flow between a designated set of important nodes. Algorithmically, it is known that the maximum $k$ for which a graph admits a Steiner strongly $k$-arc-connected orientation can be determined in polynomial time via Nash-Williams' celebrated orientation theorem~\cite{nash1960orientations}, whose constructive proof provides an explicit method to produce such an orientation, and thus allows one to compute this maximum $k$ efficiently.

Motivated by the interest in Steiner strongly $k$-arc-connected orientations, Király and Lau~\cite{kiraly2006approximate,kiraly2008approximate} proposed to study its {\it rooted counterpart}. Formally, let $G=(V,E)$ be an undirected graph, $r\in V$ a designated root vertex, and $S\subseteq V-r$ a set of terminals.  An orientation of $G$ is called a {\it Steiner rooted $k$-arc-connected orientation} if the resulting directed graph is Steiner rooted $k$-arc-connected; that is, for every $s\in S$, there exist $k$ pairwise arc-disjoint directed paths from $r$ to $s$. The problem of finding such orientations naturally generalizes several classical connectivity results in graph theory, including Menger's theorem~\cite{menger1927allgemeinen} on edge-disjoint paths and the spanning-tree packing theorems of Tutte~\cite{tutte1961problem} and Nash-Williams~\cite{nash1961edge}. The notion of Steiner rooted $k$-arc-connected orientations is also closely related to {\it Steiner tree packings}~\cite{jain2003packing,kriesell2003edge,lau2004approximate,devos2016packing}, where the goal is to find trees connecting a set of terminal vertices. Indeed, any graph having $k$ edge-disjoint Steiner trees that span $S+r$ admits a Steiner rooted $k$-arc-connected orientation, though the reverse implication does not hold. Furthermore, Király and Lau~\cite{kiraly2006approximate,kiraly2008approximate} showed that determining whether a graph admits a Steiner rooted $k$-arc-connected orientation is NP-complete, a somewhat surprising result given that rooted variants of connectivity problems are typically easier than their non-rooted counterparts. However, they posed the following open question: {\it Is the Steiner rooted $k$-arc-connected orientation problem polynomially solvable for any fixed $k$?} Remarkably, even the case $k=2$ has remained unresolved. 

A natural generalization of this framework was considered by Frank, Király, and Király~\cite{frank2003orientation}, where the connectivity requirements between vertex pairs are not restricted to a single root and a fixed set of terminals. In the {\it $R$-orientation} problem, one is given a requirement function that assigns a nonnegative integer $R(u,v)$ to every ordered pair of vertices. The task is to orient the edges so that for each pair $(u,v)$ there exist at least $R(u,v)$ pairwise arc-disjoint directed paths from $u$ to $v$. This formulation captures the Steiner rooted $k$-arc-connected orientation problem as a special case by setting $R(r,s)=k$ for each terminal $s$ and $R(u,v)=0$ otherwise. The general {\it $R$-orientation} problem was shown to be NP-hard in~\cite{frank2003orientation} even when all requirements are at most three, yet its parameterized complexity has remained open. In particular, one may ask: {\it Is the problem fixed-parameter tractable when parameterized by the total requirement $\alpha = \sum_{(u,v)\in V\times V} R(u,v)$?}

These questions form the central motivation for our work.

%%%%%%%%%%%%%%%%%%%%%%%%%%%%%%%%
\subsection{Related Work and Motivation}
\label{sec:prev}
%%%%%%%%%%%%%%%%%%%%%%%%%%%%%%%%

\paragraph{Orientation problems.} While Steiner strongly $k$-arc-connected orientations and Steiner tree packings have received considerable attention, work on the rooted variants has remained rather limited. In his seminal paper~\cite{nash1960orientations}, Nash-Williams proved that every undirected graph admits an orientation such that the directed edge-connectivity between any pair of vertices is at least the floor of half the undirected edge-connectivity between them. Specifically, this implies that if the root together with the set of terminals is Steiner $2k$-edge-connected, then the graph admits a Steiner $k$-arc-connected rooted orientation. 

Király and Lau~\cite{kiraly2006approximate,kiraly2008approximate} studied the problem in the setting of undirected hypergraphs and established approximate min-max relations for the existence of Steiner rooted orientations. They showed that if the terminals are $2k$-hyperedge-connected, then a Steiner rooted $k$-hyperarc-connected orientation exists, and obtained an analogous result for graphs with element-connectivity. Their approach relies on submodular optimization and decomposition techniques and also yields constant-factor approximations for maximizing~$k$.

\paragraph{Steiner tree packings.} The Steiner Tree Packing problem is a natural common generalization of the edge-disjoint paths problem and the edge-disjoint spanning trees problem. Given a graph $G=(V,E)$ and a terminal set $S\subseteq V$, the goal is to find the maximum number of edge-disjoint Steiner trees, each connecting all vertices of $S$. When $S=\{r,v\}$, this reduces to Menger's theorem on edge-disjoint $r$-$v$ paths, and when $S=V$, it becomes the edge-disjoint spanning trees problem characterized by Tutte~\cite{tutte1961problem} and Nash-Williams~\cite{nash1961edge}. In these classical settings, the orientation and packing formulations are equivalent.

Kriesell~\cite{kriesell2003edge} conjectured that if every edge-cut separating $S$ has size at least $2k$, then $G$ contains $k$ edge-disjoint Steiner trees connecting $S$. While this conjecture remains open in general, it is known that such graphs always admit a $k$-rooted Steiner orientation for any chosen root $r\in S$, by Nash-Williams' strong orientation theorem~\cite{nash1960orientations}; however, the reverse implication does not hold in general. Hence, although the two notions are no longer equivalent, results for Steiner orientations can be viewed as relaxations of the corresponding packing statements.

\paragraph{Steiner orientations.} In the Steiner Orientation problem, we are given an undirected graph together with a collection of ordered vertex pairs, and the goal is to orient the edges so as to maximize the number of pairs $(s,t)$ for which $t$ is reachable from $s$. In this setting, 2-edge-connected subgraphs can safely be contracted as they admit a strongly connected orientation by Robbins' theorem~\cite{robbins1939theorem}. Consequently, the problem reduces to the case where the input graph is a tree, for which deciding whether there exists an orientation that connects every pair of vertices is straightforward~\cite{hassin1989orientations}.

The general maximization variant is significantly more difficult and has been the subject of extensive study, see~\cite{elberfeld2011approximability,cygan2013steiner,hakimi1997orienting}. Extensions have also been considered, including formulations on mixed graphs \cite{arkin2002note,elberfeld2013approximation,gazmu2016improved,wlodarczyk2020parameterized,cygan2013steiner,pilipczuk2018directed,hanaka2025structural,horsch2025maximum}. To the best of our knowledge, following the work of Király and Lau~\cite{kiraly2006approximate,kiraly2008approximate}, ours is the first study to address a version of this problem under higher connectivity requirements. A key additional difficulty at higher connectivity levels is that 2-edge-connected subgraphs cannot be safely contracted, and thus the instance cannot be simplified this way. 

\paragraph{Encoding complexity.} In the multicast network coding problem, a root needs to transmit $k$ packets to a set of $t$ terminals through a communication network. Each vertex may either act as a forwarding vertex, simply relaying incoming packets, or as an encoding vertex, which combines multiple incoming packets to generate new ones. Since encoding vertices require additional computational capability and introduce delay, a central objective is to minimize their number while still achieving the desired multicast capacity. Interestingly, this problem is closely related to Steiner rooted $k$-arc-connected orientations, where the encoding vertices correspond to vertices of in- or out-degree at least 2. If we aim to bound the number of vertices in an undirected graph of minimum degree at least 3 that admits a Steiner rooted $k$-arc-connected orientation but loses this property after the deletion of any edge, then bounding the number of encoding vertices plays a role similar to bounding the number of vertices in such a minimal instance. The connection is not direct, however, since in the former case the orientation is fixed, whereas in our setting the orientation itself must be found. 

Fragouli, Soljanin, and Shokrollahi~\cite{fragouli2004network} and Fragouli and Soljanin~\cite{fragouli2006information} considered the case when the network is acyclic and $k=2$, and showed that in this special case the number of encoding vertices is bounded by the number $t$ of terminals. Tavory, Feder, and Ron~\cite{feder2003bounds} obtained similar partial results and conjectured that, in acyclic networks, the number of encoding vertices depends only on the number of packets and terminals, but not on the network size. The conjecture was settled by Langberg, Sprintson, and Bruck~\cite{langberg2006encoding}, who proved that in acyclic networks the number of encoding vertices required to achieve full capacity can be bounded by $O(k^3\cdot t^2)$ and provided explicit constructions attaining this bound. They also presented examples requiring $\Omega(k^2\cdot t)$ encoding vertices. Closing the gap between the upper bound of $O(k^3\cdot t^2)$ and the lower bound of $\Omega(k^2\cdot t)$ remains an intriguing open problem. For general networks that may contain directed cycles, they showed that the minimum number of encoding vertices can be bounded in terms of the size of a minimum feedback arc set, and that determining or approximating this number is NP-hard. Xu and Han~\cite{xu2015minimum} later studied a related problem, aiming to bound the number of high-degree non-terminal vertices, called hubs, required to satisfy multiple flow demands between distinct source-sink pairs. They proved that the minimum number of such hubs is bounded independently of the network size; however, in their model, all edges not incident to any terminal are bidirected.

%%%%%%%%%%%%%%%%%%%%%%%%%%%%%%%%
\subsection{Our Results}
\label{sec:our}
%%%%%%%%%%%%%%%%%%%%%%%%%%%%%%%%

The hardness of determining the maximum $k$ for which a Steiner rooted $k$-arc-connected orientation exists, together with the open question of Király and Lau on the polynomial-time solvability for fixed $k$, motivates a parameterized complexity approach to the problem. This direction is further supported by the general $R$-orientation problem, which is NP-hard even when all connectivity requirements are at most three, and for which the parameterized complexity with respect to the total requirement has remained open. 

In the parameterized framework, besides the input size $n$, a list of integer parameters $k_1,\ldots,k_q$ is considered, and the running time of an algorithm is analyzed as a function of both $n$ and these parameters. The main goal is to design fixed-parameter algorithms, that is, algorithms running in time $f(k_1,\ldots,k_q)\cdot n^{O(1)}$ for some computable function $f$. Although the function $f$ can grow quickly, such algorithms provide a precise understanding of how the computational difficulty depends on the parameters and are often the right notion of tractability when polynomial-time solvability is out of reach.

Our main contribution is a full characterization of the parameterized complexity of the Steiner rooted $k$-arc-connected orientation problem with respect to two key parameters: the target connectivity $k$ and the number of terminals $t$. For ease of discussion, we denote the problem by \ktori when both $k$ and $t$ are fixed. Formally, given an undirected graph $G=(V,E)$, a specified root vertex $r\in V$, and a subset of vertices $S\subseteq V-r$ of size $t$, the goal is to find an orientation of $G$ that is rooted $k$-arc-connected from $r$ to every vertex in $S$. When either $k$ or $t$ is not fixed, we omit it from the notation. That is, \tori, \kori, and \ori refer to the variants where $k$, $t$, or both $k$ and $t$ are part of the input, respectively. 

Throughout the paper, we assume without loss of generality that the input graph contains at most $2k$ parallel edges between any pair of vertices. Indeed, in any Steiner rooted $k$-arc-connected directed graph, if more than $k$ arcs go in the same direction between two vertices, removing the excess arcs preserves Steiner rooted $k$-arc-connectivity. Therefore, the total number of edges is polynomial in $k$ and the number $n$ of vertices.

Our first main theorem provides a positive answer to the natural open question for fixed $k$ and $t$, showing that \kori can indeed be solved efficiently in a parameterized sense.

\begin{restatable}{thm}{mainthm} \label{thm:main1}
    \ktori can be solved in time $f(k,t)\cdot n^{O(1)}$.
\end{restatable}

The proof is carried out in multiple steps and relies on a careful combination of fundamental results from combinatorial optimization and extremal combinatorics. Key ingredients include the Erdős-Pósa property for directed cycles, established by Reed, Robertson, Seymour, and Thomas~\cite{reed1996packing}; Ramsey's theorem on monochromatic cliques~\cite{Ramsey1930}; the sunflower lemma of Erdős and Rado~\cite{erdos1960intersection}; the Erdős-Szekeres theorem on monotone subsequences~\cite{erdos1935combinatorial}; and the fixed-parameter algorithms of Fomin, Lokshtanov, Panolan, Saurabh, and Zehavi~\cite{fomin2020hitting} for finding topological minors. It is truly remarkable that these classical and seemingly unrelated results can be combined seamlessly, ultimately leading to an efficient algorithm for the problem.

An important and somewhat unexpected consequence of our main result is that the orientation problem with local connectivity requirements is also fixed-parameter tractable when parameterized by the total requirement. In other words, our framework not only resolves the rooted Steiner case but simultaneously yields tractability for the more general setting of arbitrary pairwise connectivity demands. We denote the general problem by \rori, and by \roria when it is parameterized by the total demand $\alpha$.

\begin{restatable}{cor}{maincor}\label{cor:main}
    \roria can be solved in time $g(\alpha)\cdot n^{O(1)}$.
\end{restatable}

Despite our positive result on parameterized complexity, the rooted problem remains computationally challenging when only one of the parameters is fixed. In particular, we show the following hardness result for fixed $k$.

\begin{restatable}{thm}{kfix} \label{kfix}
    For any fixed $k \geq 2$, \kori is NP-hard.
\end{restatable}

Similarly, fixing the number of terminals $t$ does not make the problem tractable. The following theorem establishes NP-hardness for any fixed $t \ge 4$. Compared to the case of fixed $k$, the proof is considerably more intricate, thus we present it in multiple stages.

\begin{restatable}{thm}{tfix} \label{tfix}
    For any fixed $t \geq 4$, \tori is NP-hard.
\end{restatable}

Since local connectivity requirements generalize the rooted setting, these hardness results naturally extend to \rori as well.

%%%%%%%%%%%%%%%%%%%%%%%%%%%%%%%%
\subsection{Overview of Techniques}
\label{sec:techniques}
%%%%%%%%%%%%%%%%%%%%%%%%%%%%%%%%

In what follows, we give a high-level overview and a roadmap for the proofs of our main results, Theorem~\ref{thm:main1}, Corollary~\ref{cor:main}, Theorem~\ref{kfix} and Theorem~\ref{tfix}, describing the key ideas and the structure of the arguments.

%%%%%%%%%%%%%%%%%%%%%%%%%%%%%%%%
\subsubsection{Algorithmic Result}
%%%%%%%%%%%%%%%%%%%%%%%%%%%%%%%%

The proof of \cref{thm:main1} builds on the observation that for fixed $k$ and $t$, there are only finitely many feasible instances that are minimal under edge deletions. The theorem follows by generating all such minimal instances and testing whether the given instance contains one of them as a topological minor. Proving the finiteness bound is, however, far from straightforward and requires several technical tools. Corollary~\ref{cor:main} then follows from \cref{thm:main1} by a simple construction.

%%%%%%%%%%%%%%%%%%%%%%%%%%%%%%%%
\paragraph{Bounding the degrees.}
%%%%%%%%%%%%%%%%%%%%%%%%%%%%%%%%

We first reduce the problem to the case where the root and the terminals have degree exactly $k$, while all remaining vertices have degree $3$. For ease of discussion, we refer to such instances as {\it 3-regular}. Achieving the degree constraints for the root and the terminals can be done in a straightforward manner by adding new vertices and parallel edges (Lemma~\ref{lem:degk}); however, ensuring 3-regularity requires a more careful construction using binary trees (Lemma~\ref{lem:deg3}).

%%%%%%%%%%%%%%%%%%%%%%%%%%%%%%%%
\paragraph{Fixed topological minors.}
%%%%%%%%%%%%%%%%%%%%%%%%%%%%%%%%

We provide an overview of basic tools for studying $W$-fixed topological minors in both graphs and directed graphs. We prove a characterization of an undirected or directed graph being a topological minor of another with some subset of vertices being fixed (\cref{serdzftuvbhjnik}). This allows us to connect feasible orientations to fixed topological minors by proving that a $3$-regular instance is a yes-instance precisely when it contains a minimal $3$-regular yes-instance as a fixed topological minor (Lemmas \ref{lem:orintedtop} and~\ref{lem:ortopmin}). We also establish a structural lemma showing that certain vertices must appear in every fixed minor under tight-cut conditions (\cref{lem:essentialvertex}). 

%%%%%%%%%%%%%%%%%%%%%%%%%%%%%%%%
\paragraph{Feasible orientations with bounded feedback arc set.}
%%%%%%%%%%%%%%%%%%%%%%%%%%%%%%%%

We next consider minimal 3-regular instances. An orientation is {\it feasible} if it is Steiner rooted $k$-arc-connected, and a graph admitting a feasible orientation is {\it minimal} if the deletion of any edge destroys all feasible orientations. We show that for any feasible orientation of a minimal 3-regular instance, if the minimum size of a feedback arc set is small (Proposition~\ref{prop:fac}), then the number of vertices in the graph can be bounded. Otherwise, the Erdős-Pósa property for directed cycles implies that we can find a large collection of pairwise vertex-disjoint directed cycles. Some of these results appeared explicitly or implicitly in the work of Langberg, Sprintson, and Bruck~\cite{langberg2006encoding}, but due to the different setting, we provide full proofs again.

%%%%%%%%%%%%%%%%%%%%%%%%%%%%%%%%
\paragraph{Structure of directed cycles.}
%%%%%%%%%%%%%%%%%%%%%%%%%%%%%%%%

Given a Steiner rooted $k$-arc-connected directed graph containing a large collection of pairwise vertex-disjoint directed cycles, we analyze the structure of these cycles. For each terminal $s$, we consider the family of tight $s$-cuts separating $s$ from the root, and, by applying the sunflower lemma of Erdős and Rado~\cite{erdos1960intersection} and Ramsey's theorem~\cite{Ramsey1930}, identify a subcollection of cycles that captures the structure of these cuts (Proposition~\ref{prop:order}). Using a result of Erdős and Szekeres~\cite{erdos1935combinatorial} on monotone subsequences, we then find an ordered sequence $(C_1,\ldots,C_q)$ of directed cycles such that, for every terminal $s$, either $(C_1,\ldots,C_q)$ or $(C_q,\ldots,C_1)$ forms an $s$-ordered collection, meaning that the cycles appear in a consistent nested order with respect to the corresponding tight $s$-cuts (Proposition~\ref{prop:moreterminalorder}). 

%%%%%%%%%%%%%%%%%%%%%%%%%%%%%%%%
\paragraph{Inductive argument.}
%%%%%%%%%%%%%%%%%%%%%%%%%%%%%%%%

To bound the number of vertices in a minimal $3$-regular instance, we analyze how directed cycles can appear in a feasible orientation of such graphs. The key idea is to study how the terminals interact through collections of vertex-disjoint cycles and the cuts intersecting those. Depending on how these cycles and cuts are arranged, we distinguish three scenarios. In the first case, certain cuts are sparse enough to allow the construction of a smaller equivalent instance, contradicting minimality (Lemma~\ref{lem:case1}). In the second case, all terminals behave in a similar way with respect to the cycle structure, which again leads to a redundancy that violates minimality (Lemma~\ref{lem:case2}). In the third case, the terminals split into two groups that relate to the cycle structure in opposite ways, and the corresponding configuration gives rise to a directed cycle that cannot occur in a minimal instance (Lemma~\ref{lem:case3}). In all cases, a contradiction arises whenever the instance is large enough, implying that the number of vertices in a minimal $3$-regular instance is bounded by a constant depending only on $k$ and $t$.

%%%%%%%%%%%%%%%%%%%%%%%%%%%%%%%%
\paragraph{Putting everything together.}
%%%%%%%%%%%%%%%%%%%%%%%%%%%%%%%%

Finally, we apply the bound on minimal $3$-regular instances within the framework of fixed topological minors to conclude the proof of the main result. We first show that, for fixed $k$ and $t$, all non-isomorphic minimal $3$-regular instances can be enumerated efficiently (Lemma~\ref{lem:enumtopo}). Using the result of Fomin, Lokshtanov, Panolan, Saurabh, and Zehavi~\cite{fomin2020hitting} on testing fixed topological minors, we then check for each minimal pattern whether it occurs in the given instance. Since the number and size of these patterns depend only on $k$ and $t$, this yields an algorithm with running time $f(k,t)\cdot n^{O(1)}$, completing the proof of \Cref{thm:main1}.

%%%%%%%%%%%%%%%%%%%%%%%%%%%%%%%%
\subsubsection{Hardness Results}
%%%%%%%%%%%%%%%%%%%%%%%%%%%%%%%%

%%%%%%%%%%%%%%%%%%%%%%%%%%%%%%%%
\paragraph{Hardness for fixed $k$.}
%%%%%%%%%%%%%%%%%%%%%%%%%%%%%%%%

The hardness proof for \Cref{kfix} is as follows. The main technical difficulty lies in showing NP-hardness of \twoori under an additional technical condition on the instance (Lemma~\ref{serdgh}); the general case for $k \ge 3$ then follows in a straightforward way. The proof uses a reduction from \textsc{Monotone Not-All-Equal 3-Satisfiability}. Given a formula, we construct a corresponding instance of \twoori that captures the structure of the clauses and variables. The reduction ensures that satisfying assignments of the formula correspond precisely to Steiner rooted $2$-orientations in the constructed graph.

%%%%%%%%%%%%%%%%%%%%%%%%%%%%%%%%
\paragraph{Hardness for fixed $t$.}
%%%%%%%%%%%%%%%%%%%%%%%%%%%%%%%%

Due to the increased difficulty, the proof of the NP-hardness of \tori for any $t \ge 4$ is divided into three parts. In the first part, we introduce a variant of the satisfiability problem that is particularly convenient for our reduction. In \textsc{3-COL-MAX-2-SAT}, the input is a \textsc{MAX-2-SAT} instance together with a 3-coloring of the clauses, such that clauses of the same color do not share variables, and for each color class, a prescribed number of clauses must be satisfied. We show that \textsc{3-COL-MAX-2-SAT} is NP-hard (Lemma~\ref{3colhard}).

Next, we study \mthreeori, a variant of \threeori in which a highly structured subset of edges is required to be preoriented. We show that \mthreeori is NP-hard (Lemma~\ref{etxdrcdzfvzugbu}). The construction again relies on variable and clause gadgets similar to those used in the proof of \Cref{serdgh}, but here we use only one terminal per color class of clauses instead of separate terminals for each clause.

Finally, we prove the hardness of \fourori by a reduction from \mthreeori. We replace the preoriented edges with undirected ones and introduce additional gadgets to enforce the desired orientations. To achieve this, we add one extra terminal and connect it to the relevant edges in a way that forces their direction. The result then implies \Cref{tfix}.

%%%%%%%%%%%%%%%%%%%%%%%%%%%%%%%%
\subsection{Basic Notation}
\label{sec:notation}
%%%%%%%%%%%%%%%%%%%%%%%%%%%%%%%%

We give the basic definitions and notation here; all additional terminology is defined when first needed. 

We denote the sets of \textit{nonnegative} and \textit{positive integers} by $\bZ_{\geq 0}$ and $\bZ_+$, respectively. For $i,j\in\bZ_+$ with $i\le j$, we use $[i,j]\coloneqq\{i,\ldots,j\}$, and write $[j]$ when $i=1$. Given a ground set $V$, the \emph{symmetric difference} of $X,Y\subseteq V$ is denoted by $X\triangle Y\coloneqq (X\setminus Y)\cup(Y\setminus X)$. If $Y$ consists of a single element $y$, then $X\setminus \{y\}$ and $X\cup \{y\}$ are abbreviated as $X-y$ and $X+y$, respectively. 

All graphs and directed graphs considered are loopless but may contain parallel edges. For a graph $G=(V,E)$ with vertex set $V$ and edge set $E$, we write $\delta_G(X)$ for the set of edges having exactly one end vertex in $X$. For a vertex $v\in V$, its {\it degree} is $d_G(v)=|\delta_G(v)|$. In the directed case, we write $\delta_G^{in}(X)$ and $\delta_G^{out}(X)$ for the sets of arcs entering and leaving $X$, respectively, and $d_G^{in}(X)=|\delta_G^{in}(X)|$ and $d_G^{out}(X)=|\delta_G^{out}(X)|$ for the {\it in-degree} and {\it out-degree} of $X$. The set of edges {\it induced} by $X$ is denoted by $E[X]$. For a subset $F\subseteq E$, we write $V(F)$ for the set of vertices incident with edges in $F$. For a directed path $P$ and vertices $u,v\in V(P)$ with $u$ preceding $v$ along $P$, we denote by $P[u,v]$ the {\it subpath} of $P$ from $u$ to $v$. If $P$ and $Q$ are directed paths such that the last vertex of $P$ coincides with the first vertex of $Q$, then $P\circ Q$ denotes their {\it concatenation}, that is, the directed walk obtained by first taking the arcs of $P$ and then those of $Q$. For two vertices $u,v\in V$, the {\it maximum number of pairwise arc-disjoint directed $u$-$v$ paths} is denoted by $\lambda_{G}(u,v)$. For $F\subseteq E$ and $X\subseteq V$, we write $G-F$ and $G-X$ for the graphs or directed graphs obtained by deleting $F$ or $X$, respectively. In a directed graph, a {\it feedback arc set} or a {\it feedback vertex set} is a subset of arcs or vertices that intersects every directed cycle, or equivalently, whose deletion results in an acyclic directed graph.

%%%%%%%%%%%%%%%%%%%%%%%%%%%%%%%%
\subsection{Organization}
\label{sec:org}
%%%%%%%%%%%%%%%%%%%%%%%%%%%%%%%%

The rest of the paper is organized as follows. Section~\ref{sec:fixkt} proves Theorem~\ref{thm:main1} following the outline from Section~\ref{sec:techniques}: Section~\ref{sec:3reg} reduces the problem to equivalent $3$-regular instances, Section~\ref{sec:ftm} provides some basic resilts on fixed topological minors, Section~\ref{sec:feasible} rules out feasible orientations with small feedback arc sets in large minimal instances, Section~\ref{sec:ordering} establishes the existence of a long ordered sequence of directed cycles with distinguished structural properties, Section~\ref{sec:bounding} uses this ordering to bound the size of minimal instances as a function of $k$ and $t$, and Section~\ref{sec:proofofmain} applies fixed topological minor testing to obtain an $f(k,t)\cdot n^{O(1)}$ algorithm and derives Corollary~\ref{cor:main}. We note that the results of Section~\ref{sec:ordering} apply to general Steiner-rooted $k$-arc-connected directed graphs and may be of independent interest. In Section~\ref{sec:hardness}, we complement our algorithmic results by proving that the problem becomes computationally hard as soon as one of the parameters is part of the input: we show in Section~\ref{sec:fixk} that the fixed-$k$ case is NP-hard for every $k\ge 2$, answering the question of Király and Lau, and in Section~\ref{sec:fixt} that the fixed-$t$ case is NP-hard for every $t\ge 4$, with the latter proof presented in several steps due to its higher complexity. Finally, we conclude the paper with a list of open problems in Section~\ref{sec:conclusion}.

%%%%%%%%%%%%%%%%%%%%%%%%%%%%%%%%
\section{Algorithm Parameterized by \texorpdfstring{$k$}{k} and \texorpdfstring{$t$}{t}}
\label{sec:fixkt}
%%%%%%%%%%%%%%%%%%%%%%%%%%%%%%%%

In this section, we follow the outline from Section~\ref{sec:techniques} to prove Theorem~\ref{thm:main1} and Corollary~\ref{cor:main}; the detailed arguments for each step are presented in Sections~\ref{sec:3reg}–\ref{sec:proofofmain}. Throughout, we use notation of the form $f_{\text{x.y}}(I)$ to denote a function whose value depends only on the choice of the parameters in $I$, where the subscript x.y indicates the number of the statement in which it is introduced. These functions can be computed explicitly for any fixed choice of parameters and will be used later to bound the running time of our algorithm. 

%%%%%%%%%%%%%%%%%%%%%%%%%%%%%%%%
\subsection{Preprocessing the Instance}
\label{sec:3reg}
%%%%%%%%%%%%%%%%%%%%%%%%%%%%%%%%

In this section, we show that it suffices to consider graphs of a particular structure. We first show that it suffices to consider graphs in which the root and the terminals have degree exactly $k$, and all other vertices have degree exactly $3$. Recall that we refer to such instances as 3-regular.

\begin{lem}\label{lem:degk}
    \ktori admits a polynomial-time reduction to instances in which the root and all terminals have degree exactly $k$.
\end{lem}
\begin{proof}
    Let $(G=(V,E),S,r)$ be an instance of \ktori. We construct a new instance $(G'=(V',E'),S',r')$ by adding, for each $v\in S+r$, a new copy $v'$ together with $k$ parallel edges between $v$ and $v'$. We let $S'=\{v'\colon v\in S\}$ and $r'$ be the copy of $r$. Then, it is not difficult to check that any Steiner rooted $k$-arc-connected orientation of $G$ can be extended to a Steiner rooted $k$-arc-connected orientation of $G'$ by orienting the newly added edges from $r'$ to $r$ and from $s$ to $s'$ for each $s\in S$. Conversely, in any Steiner rooted $k$-arc-connected orientation of $G'$, these edges must be oriented this way; hence, restricting to the original graph $G$ yields a Steiner rooted $k$-arc-connected orientation of $G$.
\end{proof}

A {\it binary tree rooted at $v$} is a tree in which all vertices have degree $3$ or $1$, except for the root $v$, which has degree $2$. To prove that the degrees of non-root and non-terminal vertices can be bounded, we need the following simple observation.

\begin{lem}\label{lem:binary}
    For any $d\in\bZ_+$ with $d\ge 2$, there exists a rooted binary tree with $d$ leaves and $2d-1$ vertices.
\end{lem}
\begin{proof}
    A path of length two is a binary tree with two leaves and three vertices, where the root is the vertex of degree two. Suppose that for some $d \ge 2$, there exists a rooted binary tree $T$ with $d$ leaves and $2d-1$ vertices. Take any leaf of $T$ and attach two new vertices to it. The resulting tree is again a binary tree, its number of leaves increases by one, and its number of vertices increases by two, yielding a tree with $d+1$ leaves and $2(d+1)-1$ vertices. By induction, the claim holds for all $d \ge 2$.
\end{proof}

Let $G=(V,E)$ be a graph and let $v\in V$ be a vertex of degree $2$. \emph{Suppressing} $v$ means deleting $v$ and replacing its two incident edges $uv$ and $vw$ with a single edge $uw$ if $u\neq w$; in this case, we call the deleted edges $uv$ and $vw$ the \emph{parents} of the new edge $uw$. When performing several edge deletions, vertex deletions, and vertex suppressions in a graph $G_1=(V_1,E_1)$ to obtain a graph $G_2=(V_2,E_2)$, we call an edge $e\in E_1$ an \emph{ancestor} of an edge $f\in E_2$ if there exists a sequence of edges $e=e_1,e_2,\ldots,e_\ell=f$ such that $e_i$ is a parent of $e_{i+1}$ for all $i\in[\ell-1]$. It is easy to see by induction that the set of ancestors of an edge $uv\in E_2$ is precisely the set of edges forming a path between $u$ and $v$ in $G_1$. We say that $f$ is a \emph{descendant} of $e$ if $e$ is an ancestor of $f$. The notion of vertex suppression extends naturally to directed graphs. For a directed graph $D=(V,A)$ and a vertex $v\in V$ with $d_D^{in}(v)=d_D^{out}(v)=1$, \emph{suppressing} $v$ means deleting $v$ and replacing its two incident arcs $uv$ and $vw$ with a single arc $uw$ if $u\neq w$. The definitions of parents, ancestors, and descendants extend analogously to directed graphs.

Using Lemma~\ref{lem:binary}, we give a simple reduction to the 3-regular case.

\begin{lem}\label{lem:deg3}
    \ktori admits a polynomial-time reduction to equivalent instances in which every non-root and non-terminal vertex has degree $3$, while the degrees of the root and the terminals are preserved.
\end{lem}
\begin{proof}
    Let $(G=(V,E),S,r)$ be an instance of \ktori\ (Figure~\ref{fig:red1}). We construct a new instance $(G'=(V',E'),S,r)$ by replacing each vertex $v \in V \setminus (S + {r})$ with a gadget as follows. Let $v_1,\ldots,v_{d(v)}$ be new vertices, each receiving one of the original edges incident to $v$. If $d(v)=2$, then we add an edge between $v_1$ and $v_2$. If $d(v)\geq 3$ then, for each $i\in[d(v)]$, we add a binary tree $T_i$ rooted at $v_i$ with $d(v)-1$ leaves. For every distinct pair $i,j\in[d(v)]$, we add an edge between a leaf of $T_i$ and a leaf of $T_j$, using different leaves for different pairs (Figure~\ref{fig:red2}). We claim that the original instance $(G=(V,E),S,r)$ has a solution if and only if the instance $(G'=(V',E'),S,r)$ constructed as above has a solution.

    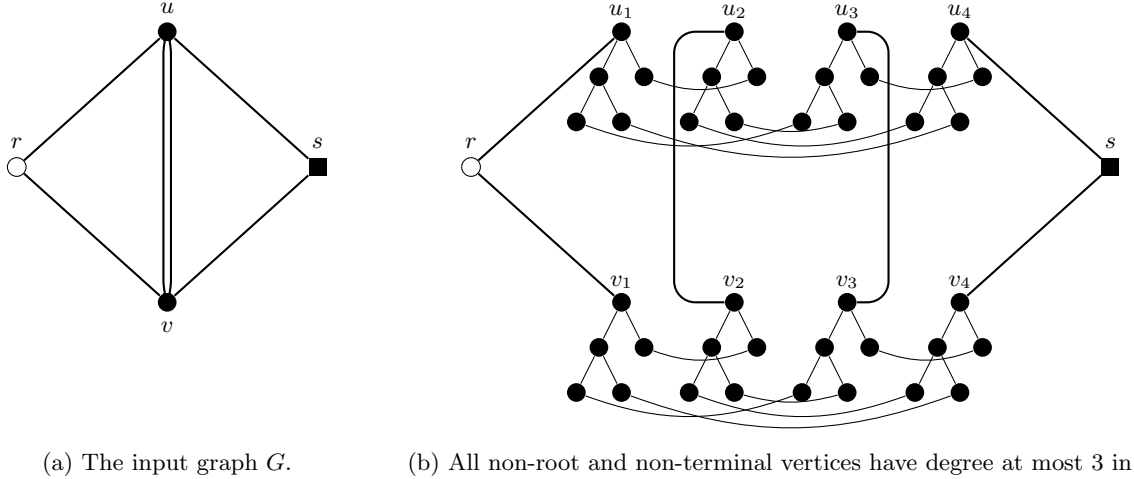
\begin{figure} \centering
        \begin{subfigure}[t]{0.35\textwidth} \centering
            \begin{tikzpicture}[every circle node/.style={inner sep=2.5pt}, font=\small]
        		\clip (-2.15,-1.75) rectangle (2.15,4.05);
                \node[circle, draw, label={$r$}] (r) at (-2,1.8) {};
        		\node[circle, fill, label={$u$}] (u) at (0, 3.6) {};
        		\node[circle, fill, label={270:$v$}] (v) at (0, 0) {};
        		\node[fill, label={$s$}] (s) at (2,1.8) {};
        		\draw[thick] (r) -- (u); 
        		\draw[thick] (r) -- (v); 
        		\draw[thick] (u) -- (s);
        		\draw[thick] (v) -- (s);  
        		\draw[thick, rounded corners=4pt] (u) -- ($(u)+(-0.05, -0.15)$) -- ($(v)+(-0.05, 0.15)$) -- (v);
        		\draw[thick, rounded corners=4pt] (u) -- ($(u)+(0.05, -0.15)$) -- ($(v)+(0.05, 0.15)$) -- (v);
        	\end{tikzpicture}
            \caption{The input graph $G$.} \label{fig:red1}
        \end{subfigure}
        \begin{subfigure}[t]{0.64\textwidth} \centering
            \newcommand{\drawbinarytree}[2]{%
            	\begin{scope}[shift={#2}]
            		\node[circle, fill, label={[yshift=-0.2em]:$#1$}] (#11) at (0,0) {};
            		\node[circle, fill] (#12) at (-0.3,-0.6) {};
            		\node[circle, fill] (#13) at (-0.6, -1.2) {};
            		\node[circle, fill] (#14) at (0,-1.2) {};
            		\node[circle, fill] (#15) at (0.3,-0.6) {}; 
            		
            		\draw (#11) -- (#12);
            		\draw (#12) -- (#13);
            		\draw (#12) -- (#14);
            		\draw (#11) -- (#15);
            	\end{scope}
            }
            \newcommand{\drawgadget}[2]{%
            	\begin{scope}[shift={#2}]
            		\drawbinarytree{#1_1}{(0,0)};
            		\drawbinarytree{#1_2}{(1.5,0)};
            		\drawbinarytree{#1_3}{(3,0)};
            		\drawbinarytree{#1_4}{(4.5,0)};
            		
            		\draw[bend right=20] (#1_15) to (#1_25);
            		\draw[bend right=20] (#1_35) to (#1_45);
            		\draw[bend right=20] (#1_13) to (#1_33);
            		\draw[bend right=20] (#1_23) to (#1_43);
            		\draw[bend right=20] (#1_14) to (#1_44);
            		\draw[bend right=15] (#1_24) to (#1_34);
            	\end{scope}
            }
        
            \begin{tikzpicture}[every circle node/.style={inner sep=2.5pt}, font=\small]
            	\clip (-2.15,-1.75) rectangle (6.65,4.05);
                \drawgadget{u}{(0,3.6)};
            	\drawgadget{v}{(0,0)};
            	\node[circle, draw, label={$r$}] (r) at (-2,1.8) {};
            	\draw[thick] (r) -- (u_11); 
            	\draw[thick] (r) -- (v_11); 
            	\draw[rounded corners=8pt, thick] (u_21) -- ($(u_21)+(-0.8,0)$) -- ($(v_21)+(-0.8,0)$) -- (v_21);
            	\draw[rounded corners=8pt, thick] (u_31) -- ($(u_31)+(0.55,0)$) -- ($(v_31)+(0.55,0)$) -- (v_31);
            	\node[fill, label={$s$}] (s) at (6.5,1.8) {};
            	\draw[thick] (u_41) -- (s);
            	\draw[thick] (v_41) -- (s);
            \end{tikzpicture}
            \caption{All non-root and non-terminal vertices have degree at most $3$ in $G'$.} \label{fig:red2}
        \end{subfigure}
        \caption{An illustration of the proof of \cref{lem:deg3} with $k=2$ and $S=\{s\}$.}
    \end{figure}

    To see this, first take a Steiner rooted $k$-arc-connected orientation of $G$. For each edge $uv\in E$ with $u,v\in V\setminus(S +r)$, let $u_iv_j$ denote the corresponding edge in $G'$. If $uv$ is oriented from $u$ to $v$, we orient $u_iv_j$ from $u_i$ to $v_j$, and orient all edges of the binary trees rooted at $u_i$ and $v_j$ towards $u_i$ and away from $v_j$, respectively. If $u=r$ or $v=s\in S$, then no gadget is created for $r$ or $s$; in this case, the corresponding edge in $G'$ directly connects $r$ or $s$ to the gadget vertex at the other endpoint, or connects $r$ and $s$ to each other if $u=r$ and $v=s$, and we orient it the same way as in $G$. Now for each $s\in S$, fix $k$ arc-disjoint directed paths from $r$ to $s$ in the resulting directed graph. If $uv$ and $vw$ are consecutive arcs used by any of these paths with $u,w\in V\setminus (S+r)$, and $u_iv_j,v_pw_q$ are the corresponding edges in $G'$, then we orient the edge that connects the binary trees of $v_j$ and $v_p$ from the tree of $v_j$ to the tree of $v_p$. Finally, all remaining edges of $G'$ are oriented arbitrarily. With this orientation, the $k$ arc-disjoint directed paths fixed for a terminal $s\in S$ naturally correspond to $k$ arc-disjoint directed paths from $r$ to $s$ in $G'$.

    For the other direction, note that any Steiner rooted $k$-arc-connected orientation of $G'$ naturally induces a Steiner $k$-rooted orientation of $G$ by contracting each vertex gadget back to a single vertex.    

    Note that the modified graph has maximum degree at most $3$ by construction. However, vertices with only one neighbor can be deleted, and degree-$2$ vertices with two neighbors can be suppressed in an arbitrary order. These operations clearly do not affect the feasibility of the instance, and any solution to the resulting graph can be easily transformed back into a solution for the previous one.
\end{proof}

%%%%%%%%%%%%%%%%%%%%%%%%%%%%%%%%
\subsection{Fixed Topological Minors}
\label{sec:ftm}
%%%%%%%%%%%%%%%%%%%%%%%%%%%%%%%%

In this section, we provide a few basic preliminaries on fixed topological minors. Given two graphs $G_1=(V_1,E_1)$ and $G_2=(V_2,E_2)$, an {\it isomorphism} between them is a bijection $\psi\colon V_1\to V_2$ such that, for all $u,v\in V_1$, the number of edges in $E_1$ joining $u$ and $v$ equals the number of edges in $E_2$ joining $\psi(u)$ and $\psi(v)$. If $\psi(w)=w$ for every $w\in W$ with $W\subseteq V_1\cap V_2$, we call $\psi$ {\it $W$-preserving}. If there is a $W$-preserving isomorphism between $G_1$ and $G_2$, then we say that $G_1$ and $G_2$ are \emph{$W$-isomorphic}. Analogously, in the directed case we can define isomorphism, $W$-preserving isomorphism, and $W$-isomorphic directed graphs.

Given two graphs $G_1=(V_1,E_1)$ and $G_2=(V_2,E_2)$ and a set $W\subseteq V_1\cap V_2$, we say that $G_2$ is a {\it $W$-fixed topological minor of $G_1$} if there exists an injective mapping $\psi\colon V_2\to V_1$ with $\psi(w)=w$ for all $w\in W$, and a collection $\mathcal{P}=\{P_e\colon e\in E_2\}$ of pairwise internally vertex-disjoint paths in $G_1$ such that for every $e=uv\in E_2$, the path $P_e$ is a $\psi(u)\psi(v)$-path in $G_1$ that is internally disjoint from $W$. If $G_1$ and $G_2$ are directed graphs, and all paths in $\mathcal{P}$ are required to be directed, we call $G_2$ a {\it $W$-fixed directed topological minor of $G_1$}. We note that this notion is often referred to as a rooted topological minor, but we use the present terminology to avoid confusion. 

It might be confusing at first that, in the definition of topological minors, we require the paths in $\mathcal{P}$ to be internally vertex-disjoint, whereas elsewhere we use arc-disjointness. We adopt this convention because we rely on a result from \cite{fomin2020hitting}, which follows the same definition, and because \cref{serdzftuvbhjnik} is more general and natural in this form. Furthermore, whenever we use topological minors, we will always assume that the vertices of $G_1$ outside $W$ have degree at most $3$, in which case internal vertex-disjointness and edge-disjointness are equivalent.

We provide a characterization of fixed topological minors in both the undirected and directed cases.

\begin{lem}\label{serdzftuvbhjnik}
\mbox{}
    \begin{enumerate}[label=(\alph*)]\itemsep0em
        \item Let $G_1=(V_1,E_1)$ and $G_2=(V_2,E_2)$ be graphs, and let $W\subseteq V_1\cap V_2$. Then $G_2$ is a $W$-fixed topological minor of $G_1$ if and only if a graph $W$-isomorphic to $G_2$ can be obtained from $G_1$ by repeatedly deleting edges and deleting or suppressing vertices not contained in $W$. \label{it:aaaa}
        \item Let $D_1=(V_1,A_1)$ and $D_2=(V_2,A_2)$ be directed graphs, and let $W\subseteq V_1\cap V_2$. Then $D_2$ is a $W$-fixed directed topological minor of $D_1$ if and only if a directed graph $W$-isomorphic to $D_2$ can be obtained from $D_1$ by repeatedly deleting edges and deleting or suppressing vertices not contained in $W$. \label{it:bbbb}
    \end{enumerate}
\end{lem}

\begin{proof}
    First we prove \ref{it:aaaa}. For the `only if' direction, suppose that $G_2$ is a $W$-fixed topological minor of $G_1$, and let $\psi \colon V_2 \rightarrow V_1$ and $\mathcal{P} = \{P_e \colon e \in E_2\}$ be the corresponding mapping and collection of paths, respectively. After deleting all edges of $G_1$ not contained in any of the paths in $\mathcal{P}$, every vertex $v \in V_1 \setminus W$ of the obtained graph that is not contained in the image of $\psi$ has degree either $0$ or $2$. Delete the degree $0$ vertices and suppress the degree $2$ vertices in any order to obtain a graph $W$-isomorphic to $G_2$.

    For the `if' direction, suppose that $G'_1=(V'_1, E'_1)$ is a $W$-isomorphic copy of $G_2$ obtained from $G_1$ by deleting edges and deleting or suppressing vertices not contained in $W$. Let $\psi \colon V_2 \to V'_1 \subseteq V_1$ be the $W$-preserving isomorphism between $G_2$ and $G'_1$. For any $uv \in E'_1$, let $P_{uv}$ be the $uv$-path formed by the ancestors of $uv$ in $G_1$. The mapping $\psi$ together with the collection of paths $\mathcal{P} = \{P_{\psi(u)\psi(v)} \colon uv \in E_2\}$ witness that $G_2$ is a $W$-fixed topological minor of $G_1$.

    The proof of \ref{it:bbbb} is analogous.
\end{proof}

Using \cref{serdzftuvbhjnik}, we relate fixed topological minors to Steiner rooted $k$-arc-connected directed graphs and Steiner rooted $k$-arc-connected orientations in the 3-regular case. A directed graph is said to be {\it minimally} Steiner rooted $k$-arc-connected if it contains $k$ pairwise arc-disjoint directed paths from the root to every terminal vertex, but loses this property upon the deletion of any arc. We further call it {\it $3$-regular} if the in-degree of the root is $0$, the out-degree of the root is $k$, the in-degree of every terminal vertex is $k$, the out-degree of every terminal vertex is $0$, and every other vertex has the sum of its in-degree and out-degree equal to $3$. Note that the underlying undirected graph of a minimally Steiner rooted $k$-arc-connected $3$-regular directed graph is not necessarily a minimal $3$-regular instance of \ktori, since it may be possible to delete an edge of $G$ while preserving the existence of a feasible orientation. We will need the following observation.

\begin{lem} \label{lem:orintedtop}
    Let $D=(V,A)$ be a 3-regular directed graph with root $r \in V$ and set of terminals $S \subseteq V-r$. Then $D$ is Steiner rooted $k$-arc-connected if and only if there exists an $(S+r)$-fixed directed topological minor $D'$ of $D$ such that $D'$ is a minimally Steiner rooted $k$-arc-connected 3-regular directed graph.
\end{lem}

\begin{proof}
    First, suppose that $D$ is Steiner rooted $k$-arc-connected. Greedily delete edges and delete or suppress vertices not contained in $S + r$ as long as the resulting graph remains Steiner rooted $k$-arc-connected. Note that suppressing vertices never changes the property of being Steiner rooted $k$-arc-connected. The instance thus obtained is $3$-regular, and it is an $(S + r)$-fixed directed topological minor of $D$ by \Cref{serdzftuvbhjnik}\ref{it:bbbb}.

   Now assume that $D' = (V', A')$ is an $(S + r)$-fixed directed topological minor of $D$ such that $D'$ is a minimally Steiner rooted $k$-arc-connected 3-regular directed graph. Then there exists an injective mapping $\psi \colon V' \rightarrow V$ and a collection $\mathcal{P} = \{P_e \colon e \in A'\}$ of internally vertex-disjoint directed paths in $D$ such that $\psi(r) = r$, $\psi(s) = s$ for all $s \in S$, and for every $e = uv \in A'$, $P_e$ is a $\psi(u)\psi(v)$-path in $D$. Let $s \in S$; then there exist $k$ pairwise arc-disjoint directed $r$-$s$ paths $P'_1, P'_2, \ldots, P'_k$ in $D'$. For each $i \in [k]$, let $P_i$ be the concatenation of all the directed paths $P_e$ where $e$ is an arc of $P'_i$. By definition, $P_1, P_2, \ldots, P_k$ are $k$ pairwise arc-disjoint directed $r$-$s$ paths in $D$, and the lemma follows.
\end{proof}

\begin{lem}\label{lem:ortopmin}
    A $3$-regular instance $(G=(V,E),S,r)$ is a yes-instance of \kori if and only if there exists an $(S+r)$-fixed topological minor $H$ of $G$ such that $(H=(U,F),S,r)$ is a minimal $3$-regular instance.
\end{lem}
\begin{proof}
    First, suppose that $(G=(V,E),S,r)$ is a $3$-regular yes-instance of \kori. Greedily delete edges and delete or suppress vertices not contained in $S + r$ as long as the resulting graph has a feasible orientation. Note that suppressing vertices never changes the property of having a feasible orientation. Then, the instance thus obtained is $3$-regular, and it is an $(S+r)$-fixed topological minor of $G$ by \Cref{serdzftuvbhjnik}\ref{it:aaaa}.

    Now let $(H = (U, F), S, r)$ be a minimal $3$-regular instance where $H$ is an $(S + r)$-fixed topological minor of $G$. Then there exists an injective mapping $\psi \colon U \rightarrow V$ and a collection $\mathcal{P} = \{P_e \colon e \in F\}$ of internally vertex-disjoint paths in $G$ such that $\psi(r) = r$, $\psi(s) = s$ for all $s \in S$, and for every $e = uv \in F$, $P_e$ is a $\psi(u)\psi(v)$-path in $G$. Let $\dH$ be a feasible orientation of $H$. We construct an orientation $\dG$ of $G$ as follows: for every $e = uv \in F$ oriented from $u$ to $v$ in $\dH$, orient all edges of $P_e$ from $\psi(u)$ to $\psi(v)$. All remaining edges of $G$ receive arbitrary orientations.  

    Let $s \in S$; then there exist $k$ pairwise arc-disjoint directed $r$-$s$ paths $P'_1, P'_2, \ldots, P'_k$ in $\dH$. For each $i \in [k]$, let $P_i$ be the concatenation of all the paths $P_e$ where $e$ is the underlying edge of an arc of $P'_i$. By definition, the corresponding edges of $P_i$ in $\dG$ form a directed path $\dP_i$ for each $i \in [k]$, and the paths $\dP_1, \dP_2, \ldots, \dP_k$ are $k$ pairwise arc-disjoint $r$-$s$ paths in $\dG$. Hence, $\dG$ is Steiner rooted $k$-arc-connected, and therefore $(G = (V, E), S, r)$ is a yes-instance.
\end{proof}

We continue with a lemma concerning the existence of vertices that must appear in every topological minor of a graph satisfying certain properties. Let $D=(V,A)$ be a Steiner rooted $k$-arc-connected directed graph with root $r\in V$ and terminal set $S\subseteq V-r$. We call a subset $U \subseteq V$ an \emph{$s$-cut} if $r \in U$, $s \notin U$, and it is a {\it tight $s$-cut} if $d^{out}_D(U)=k$.

\begin{lem}\label{lem:essentialvertex}
    Let $D_1 = (V_1, A_1)$ be a Steiner rooted $k$-arc-connected $3$-regular directed graph with root $r \in V_1$ and set of terminals $S \coloneqq \{s_1, s_2\} \subseteq V_1 - r$. Assume that a vertex $v \in V_1 \setminus (S + r)$ has two incident edges, $e$ and $f$, such that either both of them are incoming or both are outgoing, and $e$ leaves a tight $s_1$-cut while $f$ leaves a tight $s_2$-cut. Let $D_2 = (V_2, A_2)$ be an $(S + r)$-fixed topological minor of $D_1$, and let $\psi \colon V_2 \rightarrow V_1$ and $\mathcal{P} = \{P_e \colon e \in A_2\}$ be the corresponding mapping and collection of paths, respectively. Furthermore, assume that $D_2$ is a Steiner rooted $k$-arc-connected 3-regular directed graph. Then, $v$ is contained in the image of $\psi$.
\end{lem}
\begin{proof}
    Assume that $v$ is not in the image of $\psi$. By applying \cref{serdzftuvbhjnik}\ref{it:bbbb}, this implies that we can obtain a directed graph $D'_1=(V'_1, A'_1)$ by repeatedly deleting edges and deleting or suppressing vertices not in $(S+t)$, that is $(S+t)$-isomorphic to $D_2$ and $v \notin A'_1$.
    
    Throughout the process, after every edge deletion, vertex deletion, and vertex suppression, the current graph remains Steiner rooted $k$-arc-connected. As $e$ and $f$ leave a tight $s_1$-cut and a tight $s_2$-cut, respectively, this implies that the union of any $k$ pairwise edge-disjoint $r$-$s_1$ paths contains $e$, and the union of any $k$ pairwise edge-disjoint $r$-$s_2$ paths contains $f$. Consequently, the same holds for any descendants of $e$ and $f$; hence, we can never delete $e$, $f$, their descendants, or the vertices incident to them. 
    
    From these observations, by induction on the number of steps in the process, we can prove that after every step, there is a descendant of $e$ and a descendant of $f$ incident to $v$. Assume that this holds at some point, and perform one more edge deletion, vertex deletion, or vertex suppression. We cannot delete or suppress $v$, since it has two incident edges, a descendant of $e$ and a descendant of $f$, both of which are either incoming or outgoing. Moreover, as observed earlier, we cannot delete the descendants of $e$ and $f$ incident to $v$, nor the other vertices incident to these edges. Finally, since no other operation can destroy the property that a descendant of $e$ and a descendant of $f$ are incident to $v$, it follows that $v$ was never deleted or suppressed during the process, a contradiction. 
\end{proof}

%%%%%%%%%%%%%%%%%%%%%%%%%%%%%%%%
\subsection{Feasible Orientations with Small Feedback Arc Set}
\label{sec:feasible}
%%%%%%%%%%%%%%%%%%%%%%%%%%%%%%%%

The goal of this section is to prove that a minimal 3-regular instance with a sufficiently large number of vertices has no feasible orientation with a small feedback arc set. All results in this section are essentially contained in~\cite{langberg2006encoding}, either explicitly or as consequences of the arguments presented there. However, our setup is slightly different, so we include the proofs to keep the paper self-contained. 

\begin{lem}\label{lem:ess}
    Let $D=(V,A)$ be a minimally Steiner rooted $k$-arc-connected $3$-regular directed graph with root $r$ and terminal set $S\subseteq V-r$, and let $C\subseteq A$ be a directed cycle. For each $s\in S$, let $Q_s$ be an inclusionwise minimal subset of arcs that is the arc-disjoint union of $k$ directed $r$-$s$ paths. Then, 
    \begin{enumerate}[label=(\alph*)]\itemsep0em
        \item $C\subseteq \bigcup_{s\in S} Q_s$,\label{it:aaa}
        \item $C\setminus Q_s\neq\emptyset$ for all $s\in S$.\label{it:bbb}
    \end{enumerate}
\end{lem}
\begin{proof}
    For \ref{it:aaa}, note that if $e\in C\setminus\bigcup_{s\in S} Q_s$, then deleting $e$ preserves Steiner-rooted $k$-arc-connectivity of $D$, contradicting the minimality of $D$. For \ref{it:bbb}, observe that if $C\subseteq Q_s$, then $Q_s\setminus C$ is the arc-disjoint union of $k$ directed $r$-$s$ paths and possibly some directed cycles, contradicting minimality of $Q_s$. 
\end{proof}

We first show that any minimally Steiner rooted $k$-arc-connected $3$-regular directed graph with two terminals and a small feedback arc set has a bounded number of vertices, a result of independent interest.

\begin{prop}\label{prop:directed graph}
    For any positive integers $k$ and $b$, there exists a constant $f_{\ref{prop:directed graph}}(k,b)$ such that every minimally Steiner rooted $k$-arc-connected $3$-regular directed graph with two terminals and minimum feedback arc set of size at most $b$ has at most $f_{\ref{prop:directed graph}}(k,b)$ vertices.
\end{prop}
\begin{proof}
Let $D=(V,A)$ be a minimally Steiner rooted $k$-arc-connected $3$-regular directed graph with root $r\in V$ and set of terminals $S\coloneqq\{s_1,s_2\}\subseteq V-r$ and minimum feedback arc set of size at most $b$. For $i\in[2]$, let $Q_i\subseteq A$ be a minimum-size subset of arcs that is the union of $k$ pairwise arc-disjoint directed paths from $r$ to $s_i$. Note that $Q_i$ is acyclic by Lemma~\ref{lem:ess}\ref{it:bbb}. Furthermore, $Q_i$ contains all arcs incident to $r$ and $s_i$, and for every vertex $v\in V\setminus(S+r)$, $Q_i$ includes at most one incoming and one outgoing arc of $v$. By the minimality of $D$, every vertex $v\in V\setminus(S+r)$ in fact has exactly one incoming and one outgoing arc in $Q_i$.

Color an arc $e\in A$ red if $e\in Q_1\setminus Q_2$, blue if $e\in Q_2\setminus Q_1$, and purple if $e\in Q_1\cap Q_2$. Note that every arc $e\in A$ receives a color, again by the minimality of $D$. For the same reason, for any vertex $v\in V\setminus(S+r)$, either $v$ has two incoming arcs, one red and one blue, and one outgoing purple arc, or it has two outgoing arcs, one red and one blue, and one incoming purple arc.

We define three auxiliary directed graphs as follows. Let $D_0=(V,Q_0)$ be obtained from $D$ by deleting all purple arcs and reversing all red arcs, let $D_1=(V,A_1)$ be obtained from $D$ by reversing all red arcs, and let $D_2=(V,A_2)$ be obtained from $D$ by reversing all blue arcs. Each of $D_0$, $D_1$, and $D_2$ inherits the coloring of arcs from $D$, that is, every arc retains the color its corresponding arc had in the original directed graph. We now observe some properties of these auxiliary directed graphs.

\begin{cla}\label{cl:D_12}
    $D_1$ and $D_2$ are acyclic.
\end{cla}
\begin{claimproof}
    By symmetry, it suffices to prove the claim for $D_1$. Suppose, for a contradiction, that $D_1$ contains a directed cycle $C_1\subseteq A_1$. Clearly, $V(C_1)$ cannot contain $r$, $s_1$, or $s_2$, since none of these vertices has both positive in- and out-degrees in $D_1$.
    
    Let $e=uv$ be an arc of $C_1$ that is not blue. We claim that the next arc $f=vw$ of $C_1$ must be blue. If $e$ is red, then $f$ cannot be red, as no vertex $v\in V\setminus(S+r)$ in $D$ is incident with two red arcs. It also cannot be purple, since that would imply that both a red and a purple arc leave $v$ in $D$. Hence $f$ must be blue. If $e$ is purple, then $f$ cannot be purple, as no vertex $v\in V\setminus(S+r)$ in $D$ is incident with two purple arcs, and it cannot be red either, because then both a red and a purple arc would enter $v$ in $D$. Therefore, $f$ must again be blue.
    
    Consequently, every other arc of $C_1$ is blue, while the remaining arcs are red or purple, since two consecutive blue arcs cannot appear. Let $C\subseteq A$ be the set of arcs in $D$ corresponding to those of $C_1$. Recolor the arcs of $D$ as follows: arcs in $A\setminus C$ keep their original colors, while for arcs in $C$, purple arcs become red, red arcs become purple, and blue arcs become uncolored. Let $Q'_1$ and $Q'_2$ denote the sets of arcs that are red or purple, and blue or purple, respectively, under this new coloring.
    
    By construction, for each $i\in[2]$, the set $Q'_i$ contains all arcs incident with $r$ and $s_i$, and for every vertex $v\in V\setminus(S+r)$, it contains exactly one incoming and one outgoing arc. Hence each $Q'_i$ is the union of $k$ pairwise arc-disjoint directed $r$-$s_i$ paths, possibly together with some directed cycles. Since the set of uncolored arcs is nonempty, as every other arc of $C$ is uncolored, deleting them yields a smaller directed graph that is still Steiner rooted $k$-arc-connected, contradicting the minimality of $D$.
\end{claimproof}

\begin{cla}\label{cl:D'}
    $D_0$ is the union of $k$ pairwise arc-disjoint directed paths from $s_1$ to $s_2$, and its only isolated vertex is $r$.
\end{cla}
\begin{claimproof}
    Clearly, $d_{D_0}^{out}(s_1)=k$, $d_{D_0}^{in}(s_1)=0$, $d_{D_0}^{out}(s_2)=0$, $d_{D_0}^{in}(s_2)=k$, and $d_{D_0}^{in}(r)=d_{D_0}^{out}(r)=0$. Hence, it suffices to show that $D_0$ is acyclic and that for every $v\in V\setminus(S+r)$, we have $d_{D_0}^{in}(v)=d_{D_0}^{out}(v)=1$.

    If $C$ is a directed cycle in $D_0$, then the corresponding arcs also form a directed cycle in $D_1$, contradicting \cref{cl:D_12}. Therefore, since every vertex $v\in V\setminus(S+r)$ has either two incoming arcs, one red and one blue, or two outgoing arcs, one red and one blue in $D$, the claim follows.
\end{claimproof}

Decompose $Q_0$, $Q_1$, and $Q_2$ into $k$ pairwise arc-disjoint directed paths from $s_1$ to $s_2$, from $r$ to $s_1$, and from $r$ to $s_2$, respectively. We call a path $P$ a \emph{$Q_i$-path} if it is one of these $k$ arc-disjoint paths in the decomposition of $Q_i$. Let $P_i$ denote an arbitrary $Q_i$-path for each $i\in\{0,1,2\}$.

\begin{cla}\label{cl:P_012}
   Assume that $v_1,v_2\in V\setminus(S+r)$ are consecutive vertices on $P_0$ in this order. Suppose further that both $v_1$ and $v_2$ also lie on the paths $P_1$ and $P_2$. Then the following hold:
    \begin{enumerate}[label=(\alph*)]\itemsep0em
        \item $v_2$ precedes $v_1$ on $P_1$,\label{it:aa}
        \item $v_1$ precedes $v_2$ on $P_2$.\label{it:bb}
    \end{enumerate}
\end{cla}
\begin{claimproof}
    First, we prove \ref{it:aa}. Assume, for contradiction, that $v_1$ and $v_2$ appear in this order on $P_1$. Then the arcs of $P_1[v_1,v_2]$, together with the reversed arcs of $P_0[v_1,v_2]$, form a directed closed walk that contains a directed cycle, implying the existence of a directed cycle in $D_2$, which contradicts \cref{cl:D_12}.
    
    The proof of \ref{it:bb} is analogous. Assume, for contradiction, that $v_2$ and $v_1$ appear in this order on $P_2$. Then the arcs of $P_2[v_2,v_1]$, together with the arcs of $P_0[v_2,v_1]$, form a directed closed walk that contains a directed cycle, implying the existence of a directed cycle in $D_1$, again contradicting \cref{cl:D_12}.
\end{claimproof}

    We are now ready to prove the proposition. Set $f_{\ref{prop:directed graph}}(k,b)\coloneqq(2b+1)\cdot k^3+3$. Assume, for contradiction, that $|V\setminus(S+r)|>(2b+1)\cdot k^3$. Note that every vertex $v\in V\setminus(S+r)$ lies on a $Q_i$-path for each $i\in\{0,1,2\}$. Since there are $k^3$ possible triples $(P_0,P_1,P_2)$ with $P_i$ being a $Q_i$-path for $i\in\{0,1,2\}$, there must exist a triple of paths $(P_0,P_1,P_2)$ that share at least $2b+2$ common vertices in $V\setminus(S+r)$. Let $v_1,\ldots,v_{2b+2}$ denote these vertices, appearing on $P_0$ in this order. By \cref{cl:P_012}, the vertices $v_1,\ldots,v_{2b+2}$ appear in the same order on $P_2$ and in the reverse order on $P_1$.
    Since $P_2[v_j, v_{j+1}] \circ P_1[v_{j+1}, v_j]$ forms a directed closed walk, it contains a directed cycle $C_j$ in $P_2[v_j, v_{j+1}] \cup P_1[v_{j+1}, v_j]$ for each $j \in [2b+1]$. The subpaths $P_2[v_j, v_{j+1}]$ are pairwise arc-disjoint for all $j \in [2b+1]$, and similarly, the subpaths $P_1[v_{j+1}, v_j]$ are also pairwise arc-disjoint for all $j \in [2b+1]$. Hence, no arc $e \in A$ belongs to more than two directed cycles $C_j$. Consequently, no subset $A' \subseteq A$ of arcs with $|A'| = b$ can intersect every $C_j$ for $j \in [2b+1]$, contradicting the assumption that $D$ admits a feedback arc set of size at most $b$.
\end{proof}

With the help of Proposition~\ref{prop:directed graph}, we now prove the main result of this section.

\begin{prop}\label{prop:fac}
    For any positive integers  $b$, $k$, and $t$ with $t\geq 2$, there exists a constant $f_{\ref{prop:fac}}(b,k,t)$ such that the following holds. If $(G=(V,E),S,r)$ is a minimal 3-regular instance of \ktori with $|V| \geq f_{\ref{prop:fac}}(b,k,t)$, then no feasible orientation of $G$ has a feedback arc set of size at most $b$.
\end{prop}
\begin{proof}
    Let $f_{\ref{prop:fac}}(b,k,t)\coloneqq(f_{\ref{prop:directed graph}}(k,b)-3) \cdot \binom{t}{2}+t+1$. Take a minimal $3$-regular instance on at least $f_{\ref{prop:fac}}(b,k,t)$ vertices, and assume that $\dG$ is a feasible orientation of the underlying graph with a feedback arc set of size at most $b$. For any two terminals $s,s'\in S$, there exists a minimally Steiner rooted $k$-arc-connected 3-regular directed graph $D_{s,s'}=(V_{s,s'}, A_{s,s'})$ with respect to the set of terminals $\{s,s'\}$ that is an $(\{r,s,s'\})$-fixed directed topological minor of $\dG$ by \cref{lem:orintedtop}; let $\psi_{s,s'} \colon V_{s,s'}\to V$ denote the corresponding mapping. Since $\dG$ admits a feedback arc set of size at most $b$, the same holds for each $D_{s,s'}$. Hence, by \cref{prop:directed graph}, we have $|V_{s,s'}|\le f_{\ref{prop:directed graph}}(k,b)$ for every $s,s'\in S$. Consequently, there exists a vertex $v\in V$ that is not contained in the image of any of the maps $\psi_{s,s'}$.
    
    In $\dG$, the vertex $v$ has either two incoming and one outgoing arc, or vice versa. We consider the case $d_{\dG}^{in}(v)=2$ and $d_{\dG}^{out}(v)=1$; the other case is analogous. Let $e$ and $f$ be the two edges entering $v$. Since $G$ is minimal, $e$ leaves a tight $s_1$-cut and $f$ leaves a tight $s_2$-cut for some terminals $s_1, s_2 \in S$ in $\dG$. Clearly, $s_1\neq s_2$, since $e$ and $f$ cannot both belong to the union of $k$ pairwise arc-disjoint $r$-$s$ paths. However, then \cref{lem:essentialvertex} implies that $v$ is in the image of $\psi_{s_1, s_2,}$, yielding a contradiction.
\end{proof}

%%%%%%%%%%%%%%%%%%%%%%%%%%%%%%%%
\subsection{Ordering of Directed Cycles}
\label{sec:ordering}
%%%%%%%%%%%%%%%%%%%%%%%%%%%%%%%%

Let $D=(V,A)$ be a Steiner rooted $k$-arc-connected directed graph with root $r\in V$ and terminal set $S\subseteq V-r$. We say that a cut $U$ \emph{properly intersects} a directed cycle $C$ if $V(C)\cap U \neq \emptyset$ and $V(C) \setminus U \neq \emptyset$. In particular, this implies that at least one arc of $C$ leaves the set $U$. Then, we call $C$ \emph{$s$-essential} for some $s\in S$ if there exists a tight $s$-cut properly intersecting $C$. We call a sequence of directed cycles $(C_1,\ldots,C_q)$ \emph{$s$-ordered} if there exist tight $s$-cuts $U_1,\ldots,U_q$ such that $V(C_i)\subseteq U_j$ for all $i<j$, $V(C_i)\cap U_j=\emptyset$ for all $i>j$, and $U_i$ properly intersects $V(C_i)$  for every $i$. We say that the cuts $U_1, \ldots , U_q$ \emph{witness} that $(C_1, \ldots , C_q)$ form an $s$-ordered sequence of directed cycles. A set of directed cycles $\cC$ is \emph{$s$-orderable}, if there is an ordering of the members of $\cC$ in which they form an $s$-ordered collection.

The proof of the following lemma relies on standard submodularity arguments.

\begin{lem}\label{lem:intuni}
    Let $D=(V,A)$ be a Steiner rooted $k$-arc-connected directed graph with root $r\in V$ and terminal set $S\subseteq V-r$.
    \begin{enumerate}[label=(\alph*)]\itemsep0em
        \item Let $s\in S$ and $\cU\subseteq 2^V$ be a family of tight $s$-cuts. Then, both $\bigcap_{U\in\cU} U$ and $\bigcup_{U\in\cU} U$ are tight $s$-cuts. \label{it:a}
        \item Let $S'\subseteq S$ and, for $s\in S'$, let $U_s$ be a tight $s$-cut. Suppose that there exists $s_0\in S'$ such that $d^{out}_D(U_s\cup U_{s_0})\ge k$ for all $s\in S'-s_0$. Then $\bigcap_{s\in S'} U_s$ is a tight $s$-cut for all $s\in S'$. \label{it:b}
    \end{enumerate}
\end{lem}
\begin{proof}
    First we prove \ref{it:a} by induction on $|\cU|$. The statement is trivial when $|\cU|=1$, so assume that $|\cU|\ge 2$. Take an arbitrary set $U_0\in\cU$, and define $U_\cap\coloneqq\bigcap_{U\in\cU\setminus\{U_0\}}U$ and $U_\cup\coloneqq\bigcup_{U\in\cU\setminus\{U_0\}}U$. By induction, both $U_\cap$ and $U_\cup$ are tight $s$-cuts. Let $X\in\{U_\cap,U_\cup\}$. Since both $U_0 \cap X$ and $U_0 \cup X$ are $s$-cuts and $D$ is Steiner rooted $k$-arc-connected, we have $d^{out}_D(U_0 \cap X) \ge k$ and $d^{out}_D(U_0 \cup X) \ge k$. Thus, by submodularity of the out-degree function, we get
    \begin{align*}
        k+k
        &=d^{out}_D(U_0)+d^{out}_D(X)\\
        &\geq d^{out}_D(U_0\cap X)+d^{out}_D(U_0\cup X)\\
        &\geq k+k.
    \end{align*}
    This implies that $d^{out}_D(U_0\cap X)=d^{out}_D(U_0\cup X)=k$, yielding \ref{it:a}.

    Now we prove \ref{it:b}. The statement is trivial when $|S'|=1$, so assume that $|S'|\geq 2$. By assumption, $d^{out}_D(U_s\cup U_{s_0})\geq k$ for every $s\in S'-s_0$. Furthermore, since $r\in U_s\cap U_{s_0}\subseteq V\setminus\{s,s_0\}$ and $D$ is Steiner rooted $k$-arc-connected, we have $d^{out}_D(U_s\cap U_{s_0})\geq k$. Thus, we get 
    \begin{align*}
        k+k
        &=d^{out}_D(U_s)+d^{out}_D(U_{s_0})\\
        &\geq d^{out}_D(U_s\cap U_{s_0})+d^{out}_D(U_s\cup U_{s_0})\\
        &\geq k+k.
    \end{align*}
    This implies that $d^{out}_D(U_s\cap U_{s_0})=k$. Now set $W_s\coloneqq U_s\cap U_{s_0}$ for each $s\in S'-s_0$. By the above, each $W_s$ is a tight $s_0$-cut, and therefore part \ref{it:a} implies that $\bigcap_{s\in S'-s_0} W_s=\bigcap_{s\in S'} U_s$ is a tight $s_0$-cut. As the intersection contains $r$ but no terminal from $S'$, the lemma follows.
\end{proof}

Our argument relies on some fundamental results from extremal combinatorics, the first one being Ramsey's theorem~\cite{Ramsey1930}.

\begin{thm}[Ramsey]\label{thm:ramsey}
    For any positive integers $p$ and $q$, there exists a constant $R(p,q)$ such that for any coloring of the edges of the complete graph $K_{R(p,q)}$ with $q$ colors, there exists a complete subgraph $K_p$ with all edges having the same color.
\end{thm}

A set family $\cF$ is called a {\it sunflower} if there exists a set $P$, called the \emph{kernel} of the sunflower, such that $F_1 \cap F_2=P$ for any two distinct sets $F_1, F_2 \in \cF$. The following result is due to Erdős and Rado~\cite{erdos1960intersection}.

\begin{thm}[Erdős-Rado]\label{thm:sunflower}
    For any positive integers $p$ and $q$, there exists a constant $f_{\ref{thm:sunflower}}(p,q)$ such that any family of sets $\cF$ with $|\cF| \geq f_{\ref{thm:sunflower}}(p,q)$ and $|F| \leq p$ for all $F \in \cF$ contains a sunflower of size $q$.
\end{thm}

We are ready to prove one of the main results of this section.

\begin{prop} \label{prop:order}
    For any positive integers $q$ and $k$, there exists a constant $f_{\ref{prop:order}}(q,k)$ such that the following holds. Let $D=(V,A)$ be a Steiner rooted $k$-arc-connected directed graph with root $r$ and single terminal $s$. Let $\cC$ be a family of vertex-disjoint $s$-essential directed cycles in $D$ with $|\mathcal C|\ge f_{\ref{prop:order}}(q,k)$. Then, there exists $C_1,\ldots,C_q\in\cC$ such that $(C_1,\ldots,C_q)$ forms an $s$-ordered sequence of directed cycles.
\end{prop}
\begin{proof}
    Let $f_{\ref{prop:order}}(q,k)\coloneqq k\cdot f_{\ref{thm:sunflower}}(k,1+R(\max\{q,k+1\},4))$. For ease of discussion, let $m\coloneqq \max\{q,k+1\}$ and $c\coloneqq R(m,4)$. Note that if an $s$-cut $U$ properly intersects $V(C)$ for a directed cycle $C$, then at least one edge of $C$ leaves $U$. Hence a tight $s$-cut $U$ can properly intersect at most $k$ cycles in $\cC$.
    
    For a tight $s$-cut $U$, let $\cC_U\subseteq\cC$ be the set of cycles in $\cC$ that $U$ properly intersects. We know that $|\cC_U|\le k$ for all tight $s$-cuts $U$, and for every directed cycle $C\in\cC$ there exists an $s$-cut $U$ with $C\in\cC_U$. Therefore, there are at least $f_{\ref{prop:order}}(n,k)/k$ distinct subsets of $\cC$ of the form $\cC_U$ for some tight $s$-cut $U$. By \cref{thm:sunflower}, there exists a set $\cU$ of $c+1$ tight $s$-cuts such that the sets $\cC_U$ for $U\in\cU$ are distinct and form a sunflower; let $\cP\subseteq\cC$ be its kernel. There is at most one $s$-cut $U_P\in\cU$ with $\cC_{U_P}=\cP$, so by removing one cut, we may assume that $|\cU|=c$ and $\cC_U\neq \cP$ for all $U\in\cU$. Consequently, for each $U\in\cU$, there exists a directed cycle $C_U\in\cC_U$ with $C_U\notin \cP$. Hence $U$ properly intersects $V(C_U)$, and no other $U'\in\cU$ properly intersects $V(C_U)$.
    
    Let $\cU=\{U_1,\ldots,U_c\}$. We now color the edges of the complete graph $K_{[c]}$ with four colors: red, blue, green, and yellow. For $1\le i<j\le c$, define the color of edge ${i,j}$ as follows.
    \begin{enumerate}\itemsep0em
        \item If $U_i \cap V(C_{U_j})=\emptyset$ and $U_j \cap V(C_{U_i})=\emptyset$, the edge is red.
        \item If $V(C_{U_j}) \subseteq U_i$ and $U_j \cap V(C_{U_i})=\emptyset$, the edge is blue.
        \item If $U_i \cap V(C_{U_j})=\emptyset$ and $V(C_{U_i}) \subseteq U_j$, the edge is green.
        \item If $V(C_{U_j}) \subseteq U_i$ and $V(C_{U_i}) \subseteq U_j$, the edge is yellow.
    \end{enumerate}
    By \cref{thm:ramsey}, there exists $I\subseteq[c]$ with $|I|=m$ such that all edges between pairs of elements in $I$ have the same color. 

    If the color is red, let $U\coloneqq\bigcup_{i\in I}U_i$, and if it is yellow, let $U\coloneqq\bigcap_{i\in I}U_i$. In both cases, \cref{lem:intuni}\ref{it:a} implies that $U$ is a tight cut. For each $i\in I$, $U_i$ properly intersects $V(C_{U_i})$, while for all $j\ne i$, the sets $U_j$ are either disjoint from or contain $V(C_{U_i})$, depending on the color. Hence $U$ properly intersects every $V(C_{U_i})$, contradicting that a tight cut cannot properly intersect $k+1$ directed cycles from $\cC$.
    
    If the color is blue or green, let $I=\{i_1,\ldots,i_m\}$ with $i_1<\ldots<i_m$. Then the cycles $(C_{U_{i_m}},\ldots,C_{U_{i_1}})$ form an $s$-ordered collection witnessed by $U_{i_m},\ldots,U_{i_1}$ in the blue case, and $(C_{U_{i_1}},\ldots,C_{U_{i_m}})$ form one witnessed by $U_{i_1},\ldots,U_{i_m}$ in the green case, completing the proof.
\end{proof}

Our next goal is to generalize \cref{prop:order} to directed graphs with multiple terminals. Specifically, we show that if $\cC$ is a sufficiently large collection of directed cycles that are essential for all terminals, then there exist $C_1,\ldots,C_q\in\cC$ such that for every terminal $s$, either $(C_1,\ldots,C_q)$ or $(C_q,\ldots,C_1)$ is $s$-ordered. To prove this, we will use the Erdős-Szekeres theorem~\cite{erdos1935combinatorial}.

\begin{thm}[Erdős-Szekeres] \label{thm:erdosszekeres}
    For any positive integer $q$, there exists a constant $f_{\ref{thm:erdosszekeres}}(q)$ such that every sequence of real numbers $c_1,\ldots,c_{f_{\ref{thm:erdosszekeres}}(q)}$ contains an increasing or decreasing subsequence of length $q$.
\end{thm}

Instead of \cref{thm:erdosszekeres}, we will use the following simple corollary.

\begin{cor}\label{cor:erdosszekeres}
    For any positive integers $p$ and $q$, there exists a constant $f_{\ref{cor:erdosszekeres}}(p,q)$ such that the following holds. Let $A$ be a set of size $f_{\ref{cor:erdosszekeres}}(p,q)$, and let there be $p$ total orderings of $A$. Then, it is always possible to choose a subset $\{a_1,\ldots,a_q\}\subseteq A$ such that, in each of the $p$ orderings, the elements appear either in the order $(a_1,\ldots,a_q)$ or in the reverse order $(a_q,\ldots,a_1)$.
\end{cor}
\begin{proof}
    The corollary is trivial for $p=1$, with $f_{\ref{cor:erdosszekeres}}(1,q)\coloneqq q$. We proceed by induction on $p$ and claim that $f_{\ref{cor:erdosszekeres}}(p,q)\coloneqq f_{\ref{cor:erdosszekeres}}(p-1,f_{\ref{thm:erdosszekeres}}(q))$ is sufficient. Assume we are given a set $A$ with $|A|=f_{\ref{cor:erdosszekeres}}(p,q)$ and total orders $\prec_1,\ldots,\prec_p$ on $A$. By the induction hypothesis, we can find a subset $\{a_1,\ldots,a_{f_{\ref{thm:erdosszekeres}}(q)}\}\subseteq A$ that is either increasing or decreasing with respect to all of $\prec_1,\ldots,\prec_{p-1}$. For each $i$ with $1\le i\le f_{\ref{thm:erdosszekeres}}(q)$, define $c_i\coloneqq |\{a_j\colon 1\le j\le f_{\ref{thm:erdosszekeres}}(q),a_j \prec_p a_i\}|$. By \cref{thm:erdosszekeres}, there exists a subsequence $c_{i_1},\ldots,c_{i_q}$ that is increasing or decreasing. Consequently, the corresponding subset $\{a_{i_1},\ldots,a_{i_q}\}$ satisfies the conditions of the corollary.
\end{proof}

Now we are ready to generalize \cref{prop:order} to more terminals. 

\begin{prop} \label{prop:moreterminalorder}
    For any positive integers $q$, $k$, and $t$, there exists a constant $f_{\ref{prop:moreterminalorder}}(q,k,t)$ such that the following holds. Let $D=(V,A)$ be a Steiner rooted $k$-arc-connected directed graph with root $r$ and a set $S\subseteq V-r$ of $t$ terminals. Let $\cC$ be a family of vertex-disjoint directed cycles in $D$ such that $|\cC|\ge f_{\ref{prop:moreterminalorder}}(q,k,t)$ and each $C\in\cC$ is $s$-essential for every $s\in S$. Then, there exist cycles $C_1,\ldots,C_q\in\cC$ such that for every $s\in S$, either $(C_1,\ldots,C_q)$ or $(C_q,\ldots,C_1)$ forms an $s$-ordered collection of cycles.
\end{prop}
\begin{proof}
    Let $g_t(q,k,t)\coloneqq f_{\ref{cor:erdosszekeres}}(t,q)$, and for $i\in[t]$ in a decreasing order, define $g_{i-1}(q,k,t)\coloneqq f_{\ref{prop:order}}(g_i(q,k,t),k)$. Set $f_{\ref{prop:moreterminalorder}}(q,k,t)\coloneqq g_0(q,k,t)$.
    
    Let $S=\{s_1,\ldots,s_t\}$ be an arbitrary ordering of the terminals. For $0\leq i\leq t$, we construct a family $\cC_i\subseteq\cC$ with $|\cC_i|=g_i(q,k,t)$ such that $\cC_i$ is $s_\ell$-orderable for all $\ell\in[i]$. Define $\cC_0\coloneqq\cC$. Assume that $\cC_{i-1}$ has been defined for some $i\in[t]$. Applying \cref{prop:order} to $\cC_{i-1}$ with terminal $s_i$ yields a subfamily $\cC_i\subseteq\cC_{i-1}$ that is $s_\ell$-orderable for all $\ell\in[i]$. Moreover, by the definition of $g_{i-1}(q,k,t)$, we have $|\cC_i|=g_i(q,k,t)$.

    Since $\cC_t$ is $s_\ell$-orderable for all $\ell\in[t]$, each terminal defines an ordering on $\cC_t$, in which the cycles are $s_\ell$-ordered. Finally, as $g_t(q,k,t)=f_{\ref{cor:erdosszekeres}}(t,q)$, applying \cref{cor:erdosszekeres} to $\cC_t$ with these $t$ orderings yields an appropriate collection of $q$ cycles.
\end{proof}

%%%%%%%%%%%%%%%%%%%%%%%%%%%%%%%%
\subsection{Bounding the Number of Vertices}
\label{sec:bounding}
%%%%%%%%%%%%%%%%%%%%%%%%%%%%%%%%

For $k,t\in\bZ_+$, let $N_{k,t}$ denote the maximum number of vertices in a minimal 3-regular instance of \ktori. At this point, it is not even clear that $N_{k,t}$ is finite; our next goal is to show that it is. As an illustration, we first verify this when $k=1$.

\begin{lem}\label{lem:easy}
$N_{1,t}=2t$ for all $t\in\bZ_+$. 
\end{lem}
\begin{proof}
It is straightforward to check that for $k=1$, any minimal 3-regular instance $G=(V,E)$ is a tree in which all internal vertices have degree 3, and the leaves are exactly the root and the terminals. By the Handshaking lemma and $|E|=|V|-1$, we get $2\cdot (|V|-1)=3\cdot (|V|-|S|-1)+|S|+1$, hence $|V|=2\cdot |S|$. 
\end{proof}

Next, we show that a feasible orientation of a minimal 3-regular instance cannot contain a large set with in-degree at most $k-1$. The proof is highly technical, and readers may prefer to focus on the statement on a first reading. Nevertheless, this lemma plays a crucial role in the overall argument.

\begin{lem}\label{lem:case1}
    Let $(G=(V,E),S,r)$ be a minimal $3$-regular instance of \ktori\ with $k \ge 2$, and let $\dG$ be a feasible orientation. Then, for any subset $X \subseteq V$ with $r\in X$ and $\ell \coloneqq d^{out}_{\dG}(X) \le k-1$, we have $|V\setminus X|\leq N_{\ell,t}$.
\end{lem}
\begin{proof}    
    The main idea of the proof can be summarized as follows. Starting from $\dG$, we construct a smaller directed graph $\dH$ that preserves the subgraph of $\dG$ induced by $V\setminus X$. The construction uses a contraction-like operation on the vertices in $X$, refined to ensure 3-regularity. We then show that the underlying undirected graph $H$ defines a minimal 3-regular instance of the \ltori problem. Since in this case the number of vertices in $H$ is at most $N_{\ell,t}$ and $V\setminus X$ is contained in the vertex set of $H$, the claimed upper bound follows.

    Since $\dG=(V,\dE)$ is Steiner rooted $k$-arc-connected and $d^{out}_{\dG}(X)\le k-1$, we have $S\subseteq X$. We use the convention that for an edge $e\in E$, its directed counterpart is denoted by $\de$, and vice versa. For $s\in S$, let $Q_s$ be the union of $k$ arc-disjoint directed $r$-$s$ paths in $\dG$. By the minimality of $G$ and the feasibility of $\dG$, $\dG$ is minimally Steiner rooted $k$-arc-connected, so $\bigcup_{s\in S} Q_s=\dE$. 
    
    We now construct a new graph $H=(U,A)$ by first defining its directed version $\dH=(U,\dA)$ as follows. Add $r$, $S$, and $V\setminus X$ to $U$, together with all arcs in $\dE[V\setminus X]$. Fix an arbitrary ordering $S=\{s_1,\ldots,s_t\}$. For each $\de=uv\in\delta^{out}_{\dG}(X)$, let $\{s\in S \colon \de\notin Q_s\}=\{s_{i_1},\ldots,s_{i_p}\}$ with $i_1<\ldots<i_p$. Add new vertices $w^{e}_{s_{i_1}},\ldots,w^e_{s_{i_p}}$ to $U$, the directed $r$-$v$ path $P_e=\{rw^e_{s_{i_1}},w^e_{s_{i_1}}w^e_{s_{i_2}},\ldots,w^e_{s_{i_{p-1}}}w^e_{s_{i_p}},w^e_{s_{i_p}}v\}$, and arcs $w^e_{s_{i_j}}s_{i_j}$ for each $j\in[p]$. Observe that if $\{s\in S \colon \de\notin Q_s\}=\emptyset$, then a single arc from $r$ to $v$ is added. For each $\df=uv\in\delta^{in}_{\dG}(X)$, let $\{s\in S \colon f\in Q_s\}=\{s_{i_1},\ldots,s_{i_p}\}$ with $i_1<\ldots<i_p$; note that this is nonempty for all $\df\in\delta^{in}_{\dG}(X)$. Add new vertices $w^f_{s_{i_1}},\ldots,w^f_{s_{i_{p-1}}}$ to $U$, the directed $u$-$s_{i_p}$ path $P_f=\{uw^f_{s_{i_1}},w^f_{s_{i_1}}w^f_{s_{i_2}},\ldots,w^f_{s_{i_{p-1}}}s_{i_p}\}$, and arcs $w^f_{s_{i_j}}s_{i_j}$ for each $j\in[p-1]$. For ease of discussion, we set $w^f_{s_{i_p}}\coloneqq s_{i_p}$. 
    Let us emphasize the subtle difference between how arcs in $\delta^{out}_{\dG}(X)$ and $\delta^{in}_{\dG}(X)$ are treated: in the former case we consider terminals $s$ with $\de\notin Q_s$, while in the latter case, we consider those with $\df\in Q_s$. 
    
    Let $\dH=(U,\dA)$ be the directed graph obtained, and let $H=(U,A)$ be its underlying undirected graph. By construction, $r$ has out-degree $\ell$ and in-degree $0$, each terminal in $S$ has out-degree $0$ and in-degree $\ell$, and, by the minimality of $\dG$, every vertex in $U\setminus(S+r)$ has either one incoming and two outgoing arcs, or two incoming and one outgoing arc.

    \begin{cla}\label{cl:oneless}
        $\dH$ is Steiner rooted $\ell$-arc-connected.
    \end{cla}
    \begin{proof}
        For each terminal $s\in S$, we construct $\ell$ pairwise arc-disjoint directed $r$-$s$ paths in $\dH$ using two types of paths: one for each arc in $\delta^{out}_{\dG}(X)\cap Q_s$, and one for each arc in $\delta^{out}_{\dG}(X)\setminus Q_s$. 
        
        First, consider an arc $\de=uv\in\delta^{out}_{\dG}(X)\cap Q_s$, and fix a decomposition of $Q_s$ into $k$ pairwise arc-disjoint directed $r$-$s$ paths. Among these, exactly one contains $\de$; denote it by $P$. Let $\df=xy\in\delta^{in}_{\dG}(X)$ be the first arc of $P$ after $\de$ that enters $X$ again; since $S\subseteq X$, such an arc exists, and it may happen that $v=x$. Then $P_e[r,v]\circ P[v,x]\circ P_f[x,w_s^f]\circ\{w_s^f s\}$  is a directed $r$-$s$ path in $D$, where $P[v,x]$ and $\{w_s^f s\}$ are interpreted as empty if $v=x$ and if $w_s^f=s$, respectively. Repeating this construction for every arc in $\delta^{out}_{\dG}(X)\cap Q_s$ yields $|\delta^{out}_{\dG}(X)\cap Q_s|$ pairwise arc-disjoint directed $r$-$s$ paths.

        Now consider an arc $\de=uv\in\delta^{out}_{\dG}(X)\setminus Q_s$. Then $P_e[r,w_s^e]\circ\{w_s^e s\}$ is a directed $r$-$s$ path in $D$, which is arc-disjoint from all paths constructed in the previous step. Doing this for every arc in $\delta^{out}_{\dG}(X)\setminus Q_s$ gives another $|\delta^{out}_{\dG}(X)\setminus Q_s|$ such paths. 
        
        Altogether, we have $|\delta^{out}_{\dG}(X)\cap Q_s|+|\delta^{out}_{\dG}(X)\setminus Q_s|=|\delta^{out}_{\dG}(X)|=\ell$ pairwise arc-disjoint directed $r$-$s$ paths, concluding the proof of the claim.
    \end{proof}

    Our goal is to show that $(H=(U,A),S,r)$ is a minimal $3$-regular instance of \ltori. For that, we first show that edges of $H$ not induced by $V\setminus X$ cannot be deleted.

    \begin{cla}\label{cl:fixed}
        If $g\in A$ is such that $H'\coloneqq H-g$ admits a Steiner rooted $\ell$-arc-connected orientation $\dHp$ then $g\in A[V\setminus X]$. Furthermore, every edge in $A\setminus A[V\setminus X]$ has the same orientation in $\dHp$ as in $\dH$.
    \end{cla}
    \begin{proof}
        Since the root and the terminals have degree $\ell$ in $H$, no edge incident to them can be deleted. Moreover, such edges are necessarily oriented away from the root and toward the terminals in $\dHp$, respectively.

        Now take an arc $\df=uv\in\delta^{in}_{\dG}(X)$, and let $\{s\in S \colon f\in Q_s\}=\{s_{i_1},\ldots,s_{i_p}\}$ with $i_1<\ldots<i_p$. Consider the underlying edges of $P_f=\{uw^f_{s_{i_1}},w^f_{s_{i_1}}w^f_{s_{i_2}},\ldots,w^f_{s_{i_{p-1}}}s_{i_p}\}$. Since the arcs $\{w^f_{s_{i_j}}s_{i_j}\colon j\in[p-1]\}\cup\{w^f_{s_{i_{p-1}}}s_{i_p}\}$ are present in $\dHp$ and required for ensuring Steiner rooted $\ell$-arc-connectivity, the same must hold for the remaining arcs of $P_f$, namely $\{uw^f_{s_{i_1}}\}\cup\{w^f_{s_{i_j}}w^f_{s_{i_{j+1}}}\colon j\in[p-2]\}$. Hence all underlying edges of $P_f$ are present and oriented in $\dHp$ as in $\dH$.
        
        Finally, consider an arc $\de=uv\in\delta^{out}_{\dG}(X)$, and let $\{s\in S \colon \de\notin Q_s\}=\{s_{i_1},\ldots,s_{i_p}\}$ with $i_1<\ldots<i_p$. Consider the underlying edges of $P_e=\{rw^e_{s_{i_1}},w^e_{s_{i_1}}w^e_{s_{i_2}},\ldots,w^e_{s_{i_{p-1}}}w^e_{s_{i_p}},w^e_{s_{i_p}}v\}$. Suppose that there exists an edge $g$ among them that is either missing or oriented in $\dHp$ in the opposite direction than in $\dH$. Let $\dg$ be the first such arc on $P_e$ when traversing from $r$ to $v$, and let $w^e_{s_{i_q}}$ denote its tail. Then the set $\{r\}\cup\{w^e_{s_{i_j}}\colon j\in[q]\}\cup\{s_{i_j}\colon j\in[q]\}$ has out-degree $\ell-1$ in $\dHp$, contradicting the fact that $\dHp$ is Steiner rooted $\ell$-arc-connected. Therefore, all underlying edges of $P_e$ must be present and oriented in $\dHp$ as in $\dH$.
    \end{proof}

    A useful corollary of Claim~\ref{cl:fixed} is the following.

    \begin{cla}\label{cl:iwurfb}
        Let $g\in A[V\setminus X]$ be such that $H'\coloneqq H-g$ admits a Steiner rooted $\ell$-arc-connected orientation $\dHp$. For any $s\in S$, let $Q'_s$ be the union of $\ell$ arc-disjoint directed $r$-$s$ paths in $\dHp$. Then $Q'_s\cap \dA[V\setminus X]$ consists of pairwise arc-disjoint paths that provide a one-to-one correspondence between the heads of arcs in $\delta^{out}_{\dG}(X)\cap Q_s$ and the tails of arcs in $\delta^{in}_{\dG}(X)\cap Q_s$: each starts at the head of a unique $\de\in\delta^{out}_{\dG}(X)\cap Q_s$ and ends at the tail of a unique $\df\in\delta^{in}_{\dG}(X)\cap Q_s$.
    \end{cla}
    \begin{proof}
        The statement immediately follows from Claim~\ref{cl:fixed}.
    \end{proof}

    We are now ready to prove the minimality of $H$.
    
    \begin{cla}\label{cl:min}
        $(H=(U,A),S,r)$ is a minimal $3$-regular instance of \ltori.
    \end{cla}
    \begin{proof}
        Suppose indirectly that $g\in A$ is such that $H'\coloneqq H-g$ admits a Steiner rooted $\ell$-arc-connected orientation $\dHp$. By Claim~\ref{cl:fixed}, $g\in A[V\setminus X]$. Let $\dGp$ denote the orientation of $G-g$ where edges in $E[V\setminus X]-g$ are oriented as in $\dHp$, and all remaining edges are oriented as in $\dG$. We claim that the directed graph $\dG'$ thus obtained is Steiner rooted $k$-arc-connected, contradicting the minimality of the original instance.

        To see this, take an arbitrary $s\in S$ and let $Q'_s$ be the union of $\ell$ arc-disjoint directed $r$-$s$ paths in $\dHp$. Then, by Claim~\ref{cl:iwurfb}, the set $(Q_s\setminus\dE[V\setminus X])\cup (Q'_s\cap\dA[V\setminus X])$ is the union of $k$ pairwise arc-disjoint directed $r$-$s$ paths in $\dGp$, possibly together with some directed cycles.  
    \end{proof}

    By Claim~\ref{cl:min} and by construction, we get $|V\setminus X|\leq |U|\leq N_{\ell,t}$, concluding the proof of the lemma.
\end{proof}

Finally, we give two technical lemmas, focusing on different scenarios

\begin{lem}\label{lem:case2}
    Let $(G=(V,E),S,r)$ be a minimal 3-regular instance of \ktori, and let $\dG=(V,\dE)$ be a feasible orientation. Let $X\subseteq V$ be $s$-tight for every $s\in S'$, where $S'\subseteq S$ is non-empty. Then $X$ contains no directed cycle that is essential only for terminals in $S'$.
\end{lem}
\begin{proof}
    Suppose for a contradiction that $X$ contains a directed cycle $C$ that is essential only for terminals in $S'$. For each $s\in S'$, let $Q_s$ be a minimum-size set of arcs that is the union of $k$ arc-disjoint directed paths from $r$ to $s$. Since $X$ is $s$-tight for every $s\in S'$, we may assume that the set $Q\coloneqq Q_s\cap \dE[X]$ is the same for all $s\in S'$. By Lemma~\ref{lem:ess}\ref{it:aaa} and the indirect assumption, we get $C\subseteq Q$, contradicting Lemma~\ref{lem:ess}\ref{it:bbb}.
\end{proof}

\begin{lem}\label{lem:case3}
    Let $(G=(V,E),S,r)$ be a minimal 3-regular instance of \ktori, and let $\dG=(V,\dE)$ be a feasible orientation. Let $X_1\subseteq V$ be $s$-tight for every $s\in S_1$ and $X_2\subseteq V$ be $s$-tight for every $s\in S_2$, where $S_1,S_2\subseteq S$ are disjoint and non-empty. Then $X_1\cap X_2$ contains no directed cycle that is essential only for terminals in $S_1\cup S_2$.
\end{lem}
\begin{proof}
    Let $X \coloneqq X_1 \cap X_2$. Suppose for a contradiction that $X$ contains a directed cycle $C$ that is essential only for terminals in $S_1\cup S_2$. For each $s\in S_1\cup S_2$, let $Q_s$ be a minimum-size set of arcs that is the union of $k$ arc-disjoint directed paths from $r$ to $s$. For $i\in[2]$, since $X_i$ is $s$-tight for every $s\in S_i$, we may assume that the set $Q_s\cap \dE [X]$ is the same for all $s\in S_i$ -- let $Q_i$ denote this set. By the indirect assumption and Lemma~\ref{lem:ess}\ref{it:aaa}, $C\subseteq Q_1\cup Q_2$, while Lemma~\ref{lem:ess}\ref{it:bbb} yields $C\setminus Q_i\neq\emptyset$ for $i\in[2]$. 
    
    We color each arc $e\in C$ red if $e\in Q_1\setminus Q_2$, blue if $e\in Q_2\setminus Q_1$, and purple if $e\in Q_1\cap Q_2$. By the above observations, every arc of $C$ is colored. By the $3$-regularity of the instance, no vertex in $V(C)$ may have an incoming red arc and an outgoing blue arc, nor an incoming blue arc and an outgoing red arc. Hence, along $C$, the maximal monochromatic intervals alternate so that every other interval is purple, and the remaining ones are red or blue. In particular, there exists at least one purple arc $e\in C$.

    Let $\dG' = (V, \dE')$ be the orientation of $G$ that differs from $\dG$ exactly on the edges of $C$, and let $C'$ denote the directed cycle in $\dG'$ consisting of the same underlying edges as $C$. Note that for any $X \subseteq V$, we have $d_{\dG}^{\text{out}}(X) = d_{\dG'}^{\text{out}}(X)$; hence, since $\dG$ is feasible, $\dG'$ is also feasible. Color an edge $f \in C'$ red, blue, or leave it uncolored if the corresponding edge in $C$ is blue, red, or purple, respectively.  Let $Q'_1$ and $Q'_2$ denote the sets of arcs that are red and blue, respectively, in $C'$. Let $R'_s \subseteq \dE'$ be the set of edges in $\dG'$ corresponding to the edges of $Q_s \setminus Q_i$ for all $s \in S_i$, $i \in [2]$. Define $Q'_s \coloneqq R'_s \cup Q'_i$ in $\dG'$ for all $s \in S_i$, $i \in [2]$.
 
    By construction, each $Q'_s$ is the union of $k$ arc-disjoint directed $r$–$s$ paths in $\dG'$, possibly together with some directed cycles. Let $e' \in C'$ be the arc corresponding to $e$. Observe that no $Q'_s$ contains $e'$ for any $s \in S_1 \cup S_2$. Furthermore, since $C$ is essential only for terminals in $S_1 \cup S_2$ in $\dG$, the arc $e'$ is contained in no tight $s$-cut for any $s \in S \setminus (S_1 \cup S_2)$ in $\dG'$. Therefore, deleting $e'$ from $\dG'$ results in a Steiner rooted $k$-arc-connected directed graph, contradicting the minimality of the instance.
\end{proof}

We further need the following result of Reed, Robertson, Seymour, and Thomas~\cite{reed1996packing} that establishes the so-called Erdős-Pósa property of directed cycles.

\begin{thm}[Reed, Robertson, Seymour, Thomas]\label{thm:ep}
    For every positive integer $\gamma$, there exists a constant $f_{\ref{thm:ep}}(\gamma)$ such that if the minimum feedback vertex set of a directed graph has size greater than $f_{\ref{thm:ep}}(\gamma)$, then the graph contains $\gamma$ pairwise vertex-disjoint directed cycles.
\end{thm}

Putting everything together, we now obtain the key structural result of this section: a bound on the size of minimal $3$-regular instances. This bound will play a central role in our approach, as it enables the application of existing algorithms for detecting topological minors.

\begin{prop}\label{prop:bound}
     For any positive integers $k$ and $t$, there exists a constant $f_{\ref{prop:bound}}(k,t)$ satisfying $N_{k,t}\leq f_{\ref{prop:bound}}(k,t)$, that is, every minimal 3-regular instance of \ktori has at most $f_{\ref{prop:bound}}(k,t)$ vertices.
\end{prop}
\begin{proof}
    We prove the statement by induction on $k$. When $k=1$, the proposition follows from Lemma~\ref{lem:easy}, thus we may assume that $k\ge 2$. For ease of discussion, let $N \coloneqq f_{\ref{prop:bound}}(k-1,t)+1$, and define $f_{\ref{prop:bound}}(k,t)\coloneqq f_{\ref{prop:fac}}(f_{\ref{thm:ep}}(2^t\cdot f_{\ref{prop:moreterminalorder}}(2N+3,k,t)),k,t)$. 
    
    Let $(G=(V,E),S,r)$ be a minimal 3-regular instance of \ktori on at least $f_{\ref{prop:bound}}(k,t)$ vertices, and take an arbitrary feasible orientation $\dG$ of $G$. By Proposition~\ref{prop:fac}, the minimum size of a feedback arc set is greater than $f_{\ref{thm:ep}}(2^t\cdot f_{\ref{prop:moreterminalorder}}(2N+3,k,t))$. Since $G$ is a minimal $2$-regular instance, every non-root and non-terminal vertex has either one incoming and two outgoing arcs, or two incoming and one outgoing arcs in $\dG$. In such a directed graph, the minimum size of a feedback vertex set is equal to that of a feedback arc set, since deleting a vertex has the same effect as deleting its unique incoming arc if its in-degree is $1$, or its unique outgoing arc otherwise. Consequently, the minimum feedback vertex set in $\dG$ is also greater than $f_{\ref{thm:ep}}(2^t\cdot f_{\ref{prop:moreterminalorder}}(2N+3,k,t))$. Therefore, by Theorem~\ref{thm:ep}, there exist $2^t\cdot f_{\ref{prop:moreterminalorder}}(2N+3,k,t)$ pairwise vertex-disjoint directed cycles in $\dG$. Each of these directed cycles is $s$-essential for at least one terminal $s\in S$. Thus there exists a set $\emptyset\neq S_\cC\subseteq S$ and a collection $\cC$ of at least $f_{\ref{prop:moreterminalorder}}(2N+3,k,t)$ directed cycles such that every $C\in\cC$ is $s$-essential exactly for the terminals in $S_\cC$. By Proposition~\ref{prop:moreterminalorder}, there exist $C_1,\ldots,C_{2N+3}\in\cC$ such that either $(C_1,\ldots,C_{2N+3})$ or $(C_{2N+3},\ldots,C_1)$ forms an $s$-ordered collection of cycles for all $s\in S_\cC$. Let $S_1$ denote the terminals in $S_\cC$ for which the former holds, and let $S_2$ denote those for which the latter holds. For each $s\in S_\cC$, let $U^s_1,\ldots,U^s_{2N+3}$ be the witness tight $s$-cuts, where $U^s_i$ properly intersects $C_i$.
    \medskip

    \noindent \textbf{Case 1.} There exist $s,s'\in S_1$ such that $d^{out}_{\dG}(U^s_{N+3}\cup U^{s'}_{N+3})\leq k-1$, or $s,s'\in S_2$ such that $d^{out}_{\dG}(U^s_{N+1}\cup U^{s'}_{N+1})\leq k-1$.
    \smallskip
    
    \noindent Assume that $s,s'\in S_1$; the case $s,s'\in S_2$ is symmetric. Let $X\coloneqq U^{s}_{N+3}\cup U^{s'}_{N+3}$ and set $\ell\coloneqq d^{out}_{\dG}(X)$. Since $V(C_i)\subseteq V\setminus X$ for all $i\in[N+4,2N+3]$ and the cycles are vertex-disjoint, we have $|V\setminus X|\ge N>f_{\ref{prop:bound}}(k-1,t)\geq f_{\ref{prop:bound}}(\ell,t)\geq N_{\ell,t}$. On the other hand, $d^{out}_{\dG}(X)\le k-1$ and the feasibility of the orientation together imply $S\subseteq X$. By Lemma~\ref{lem:case1}, this gives $|V\setminus X|\le N_{\ell,t}$, a contradiction.
    \medskip
    
    \noindent \textbf{Case 2.} Either $S_1=S_\cC$ and $d^{out}_{\dG}(U^s_{N+3}\cup U^{s'}_{N+3})\ge k$ for all $s,s'\in S_\cC$, or $S_2=S_\cC$ and $d^{out}_{\dG}(U^s_{N+1}\cup U^{s'}_{N+1})\ge k$ for all $s,s'\in S_\cC$.
    \smallskip
    
    \noindent Assume that $S_1=S_\cC$; the case $S_2=S_\cC$ is symmetric. Lemma~\ref{lem:intuni}\ref{it:b} implies that $X\coloneqq\bigcap_{s\in S_\cC}U^s_{N+3}$ is $s$-tight for every $s\in\cS_\cC$. By Lemma~\ref{lem:case2}, $X$ contains no directed cycle that is essential only for terminals in $S_\cC$. However, $V(C_i)\subseteq X$ for all $i\in[N+2]$, and these cycles are essential exactly for $S_\cC$, yielding a contradiction.
    \medskip
    
    \noindent \textbf{Case 3.} None of $S_1$ and $S_2$ is empty, $d^{out}_{\dG}(U^s_{N+3}\cup U^{s'}_{N+3})\ge k$ for all $s,s'\in S_1$, and $d^{out}_{\dG}(U^s_{N+1}\cup U^{s'}_{N+1})\ge k$ for all $s,s'\in S_2$.
    \smallskip
    
    \noindent Lemma~\ref{lem:intuni}\ref{it:b} implies that $X_1\coloneqq\bigcap_{s\in S_1}U^s_{N+3}$ is $s$-tight for every $s\in S_1$, and $X_2\coloneqq\bigcap_{s\in S_2}U^s_{N+1}$ is $s$-tight for every $s\in S_2$. By Lemma~\ref{lem:case3}, $X_1\cap X_2$ contains no cycle that is essential only for terminals in $S_\cC$. However, $V(C_{N+2})\subseteq X_1\cap X_2$, and $C_{N+2}$ is essential exactly for $S_\cC$, giving a contradiction.\\

    Let us mention that in Cases 2 and 3, if $|S_1|=1$ or $|S_2|=1$, there are no pairs of distinct terminals the given set, so the condition in question holds trivially. Since we arrived to a contradiction in all three cases, the number of vertices of $G$ is less than $f_{\ref{prop:bound}}(k,t)$, finishing the proof.
\end{proof}

%%%%%%%%%%%%%%%%%%%%%%%%%%%%%%%%
\subsection{Proof of Theorem~\ref{thm:main1}}
\label{sec:proofofmain}
%%%%%%%%%%%%%%%%%%%%%%%%%%%%%%%%

In this section, we derive \Cref{thm:main1} from \Cref{prop:bound}. We first introduce some additional notation. Given two graphs $G_1=(V_1,E_1)$ and $G_2=(V_2,E_2)$, an {\it isomorphism} between them is a bijection $\psi\colon V_1\to V_2$ such that, for all $u,v\in V_1$, the number of edges in $E_1$ joining $u$ and $v$ equals the number of edges in $E_2$ joining $\psi(u)$ and $\psi(v)$. If $\psi(s)=s$ for every $s\in S$ with $S\subseteq V_1\cap V_2$, we call $\psi$ {\it $S$-preserving}. Two instances $(G_1=(V_1,E_1),S,r)$ and $(G_2=(V_2,E_2),S,r)$ of \ktori are said to be {\it isomorphic} if there exists an $(S+r)$-preserving isomorphism between $G_1$ and $G_2$.

Using these definitions, we first show that \Cref{prop:bound} enables the efficient enumeration of a collection of \ktori instances, sufficient for the purposes of our algorithm.

\begin{lem}\label{lem:enumtopo}
    For any positive integers $k$ and $t$, there exists a constant $f_{\ref{lem:enumtopo}}(k,t)$ such that \ktori has at most $f_{\ref{lem:enumtopo}}(k,t)$ non-isomorphic minimal $3$-regular instances. Moreover, all non-isomorphic minimal $3$-regular instances can be enumerated in constant time for fixed $k$ and $t$. 
\end{lem}
\begin{proof}
    By \Cref{prop:bound}, every minimal $3$-regular instance has at most $f_{\ref{prop:bound}}(k,t)$ vertices. For fixed $k$ and $t$, there are only finitely many graphs on at most $f_{\ref{prop:bound}}(k,t)$ vertices in which the root and the terminals have degree $k$, and all other vertices have degree $3$. Among these, only finitely many are minimal instances. Let $f_{\ref{lem:enumtopo}}(k,t)$ denote this finite number. Hence, \ktori admits at most $f_{\ref{lem:enumtopo}}(k,t)$ non-isomorphic minimal $3$-regular instances.
    
    To enumerate all such instances, we generate all graphs on at most $f_{\ref{prop:bound}}(k,t)$ vertices, with the root $r$ and the terminals $S$ being fixed, satisfying the required degree pattern. For each such graph, we can determine in constant time whether it admits a feasible orientation and whether it is minimal, that is, whether the deletion of any edge destroys all feasible orientations. Since all graphs have constant size, $(S+r)$-preserving isomorphism can also be tested in constant time. Consequently, the overall enumeration takes constant time for fixed $k$ and $t$.
 \end{proof}

To apply the theory of fixed topological minors, we rely on the following result, which follows from \cite[Theorem~10]{fomin2020hitting}.

\begin{thm}[Fomin, Lokshtanov, Panolan, Saurabh, Zehavi] \label{thm:fixtopmin}
    Let $G=(V,E)$ and $H=(U,F)$ be graphs, and let $W\subseteq V\cap U$. Then, it can be decided in time $f_{\ref{thm:fixtopmin}}(|U|)\cdot |V|^{O(1)}$ whether $H$ is a $W$-fixed topological minor of $G$.
\end{thm}

We are now ready to prove Theorem~\ref{thm:main1}, which we restate here for convenience.

\mainthm*

\begin{proof}
Let $(G=(V,E),S,r)$ be an instance of \ktori. We first apply Lemmas~\ref{lem:degk}, and~\ref{lem:deg3} to obtain an equivalent $3$-regular instance $(G'=(V',E'),S,r)$. Next, using \Cref{lem:enumtopo}, we compute the collection $\cH$ of all non-isomorphic minimal $3$-regular instances for the given parameters $k$ and $t$, where $|\cH|\le f_{\ref{lem:enumtopo}}(k,t)$. For each $(H=(U,F),S,r)\in\cH$, we apply \Cref{thm:fixtopmin} to test whether $G'$ contains $H$ as an $(S+r)$-fixed topological minor. If such an $H$ is found, we conclude that $(G=(V,E),S,r)$ is a yes-instance and terminate the algorithm. Otherwise, after all instances in $\cH$ have been tested, we report that $(G=(V,E),S,r)$ is a no-instance.

The correctness of the algorithm follows from the equivalence of the instances $(G=(V,E),S,r)$ and $(G'=(V',E'),S,r)$ together with Lemma~\ref{lem:ortopmin}. For the running time, let $n\coloneqq |V|$. First note that $G'$ can be constructed in time $n^{O(1)}$ due to Lemmas~\ref{lem:degk}, and~\ref{lem:deg3}. By \Cref{lem:enumtopo}, the collection $\cH$ can be computed in constant time when $k$ and $t$ are fixed. Furthermore, by \Cref{thm:fixtopmin}, for every $(H=(U,F),S,r)\in \cH$, since $|U|\le f_{\ref{prop:bound}}(k,t)$ by \Cref{prop:bound}, we can test in time $f_{\ref{thm:fixtopmin}}(f_{\ref{prop:bound}}(k,t))\cdot |V'|^{O(1)}$ whether $H$ is an $(S+r)$-fixed topological minor of $G'$. As there are at most $f_{\ref{lem:enumtopo}}(k,t)$ instances in $\cH$, the total running time of the algorithm is $f(k,t)\cdot n^{O(1)}$.
\end{proof}

\begin{rem}
    As a consequence of Proposition~\ref{prop:fac}, any minimal Steiner rooted $k$-arc-connected $3$-regular directed graph with $t$ terminals and a feedback arc set of size at most $b$ has a bounded number of vertices, where the bound depends only on $b$, $k$, and $t$. At the same time, Proposition~\ref{prop:bound} states that any minimal $3$-regular instance of \ktori has a bounded number of vertices depending only on $k$ and $t$. This naturally raises the question of whether the latter result extends to the directed setting, that is, whether every minimal Steiner rooted $k$-arc-connected $3$-regular directed graph with $t$ terminals has a bounded number of vertices depending only on $k$ and $t$, independently of the size of its feedback arc set. However, this is not the case even for two terminals; see~\cite[Figure~9]{langberg2006encoding} for an example.
\end{rem}

Using a simple construction, we get an analogous result for $R$-orientations as a corollary.

\maincor*

\begin{proof}
    Consider an instance of \roria, that is, let $G=(V,E)$ be an undirected graph, $R\colon V\times V\to\bZ_{\geq 0}$ a requirement function, and $\alpha\coloneqq\sum_{(u,v)\in V\times V}R(u,v)$. Let $\cP=\{(u,v)\in V\times V\colon R(u,v)>0\}$. We construct a new graph by extending $G$ as follows. For each ordered pair $(u,v)\in\cP$, we add vertices $r_{(u,v)}$ and $s_{(u,v)}$, connect them to $u$ and $v$ with $R(u,v)$ parallel edges, respectively, and connect $r_{(u,v)}$ to every $s_{(x,y)}$ with $R(u,v)$ parallel edges whenever $(x,y)\in\cP$ and $(x,y)\neq(u,v)$. Finally, we introduce a new vertex $r$ and connect it to every $r_{(u,v)}$ with $R(u,v)$ parallel edges for each $(u,v)\in\cP$.  
    
    Let $G'=(V',E')$ denote the graph constructed above, and set $S\coloneqq\{s_{(u,v)}\colon (u,v)\in\cP\}$. It is straightforward to verify that $(G,R)$ is a yes-instance of \roria if and only if $(G',S,r)$ is a yes-instance of {\sc Steiner Rooted $\alpha$-orientation with $|\cP|$ Terminals}. By Theorem~\ref{thm:main1}, the latter can be solved in time $f(\alpha,|\cP|)\cdot n^{O(1)}$. Since $|\cP|\le\alpha$, this yields an algorithm with running time $g(\alpha)\cdot n^{O(1)}$.
    \end{proof}

%%%%%%%%%%%%%%%%%%%%%%%%%%%%%%%%
\section{Hardness Results}
\label{sec:hardness}
%%%%%%%%%%%%%%%%%%%%%%%%%%%%%%%%

The previous section presented an algorithm for the case when both $k$ and $t$ are fixed. In this section, we show that as soon as one of these parameters becomes part of the input, the problem becomes computationally hard. In \Cref{sec:fixk}, we prove that \kori is NP-hard for any fixed $k \ge 2$. In \Cref{sec:fixt}, we establish that \tori is NP-hard for any fixed $t \ge 4$; since the proof for fixed $t$ is considerably more intricate than for fixed $k$, it is presented in multiple stages. For the sake of presentation clarity, throughout this section we refer to the vertex and edge sets of a graph $G$ as $V(G)$ and $E(G)$, respectively. Similarly, for a directed graph $D$, we use $V(D)$ and $A(D)$ for its vertex set and arc set, respectively.

%%%%%%%%%%%%%%%%%%%%%%%%%%%%%%%%
\subsection{Hardness for Fixed \texorpdfstring{$k$}{k} with \texorpdfstring{$t$}{t} as Part of the Input}
\label{sec:fixk}
%%%%%%%%%%%%%%%%%%%%%%%%%%%%%%%%

The goal of this section is to prove \Cref{kfix}. For our reduction, we need to consider the problem \textsc{Monotone Not-All-Equal 3-Satisfiability} (\textsc{MNAE-3-SAT}). An instance of \textsc{MNAE-3-SAT} consists of a set of binary variables $X$ and a set of clauses $\mathcal{C}$ such that every $C \in \mathcal{C}$ contains exactly 3 nonnegated variables from $X$. An assignment $\phi\colon X\rightarrow \{\true,\false\}$ is called {\it satisfying} if for every $C \in \mathcal{C}$, there exist two variables $x,x'\in C$ with $\phi(x)\neq \phi(x')$. The problem is to decide whether there exists a satisfying assignment. We use the following well-known result~\cite{schaefer1978complexity}.

\begin{prop}[Schaefer]\label{vfcds}
    \textsc{MNAE-3-SAT} is NP-hard.
\end{prop}

The main technical challenge of the proof lies in the following result, which resolves \Cref{kfix} for $k=2$, even under a mild additional technical assumption.

\begin{lem}\label{serdgh}
    \twoori is NP-hard, even when restricted to instances $(G,S,r)$ satisfying $d_G(S)+|E(G[S])|=2|S|$.
\end{lem}
\begin{proof}
    We prove the result by a reduction from \textsc{MNAE-3-SAT}. Let $(X,\mathcal{C})$ be an instance of \textsc{MNAE-3-SAT}. We now create an instance $(G,S,r)$ of \twoori. We let $V(G)$ contain the vertex $r$, a set $Z$ consisting of two vertices $z_1$ and $z_2$, a set $U$ consisting of two vertices $u_x^1$ and $u_x^2$ for every $x \in X$, and the set $S$ that consists of a vertex $s_x$ for every $x \in X$ and a vertex $s_{x,C}^{i}$ for every $C \in \mathcal{C}$, $x \in C$ and $i \in [2]$. For every $C \in \mathcal{C}$ and every $i \in [2]$, we let $S_C^{i}$ be the set consisting of the vertices $s_{x,C}^{i}$ for $x \in C$. 
    
    We now describe the edge set of $G$. We first let $E(G)$ contain an edge connecting $r$ and $z_1$ and an edge connecting $r$ and $z_2$. Next, for every $i \in [2]$ and every $x \in X$, we let $E(G)$ contain an edge connecting $z_i$ and $u_x^{i}$. We also connect $u_x^{1}$ and $u_x^{2}$ with an edge for all $x \in X$. Further, for every $i \in [2]$ and every $x \in X$, we let $E(G)$ contain an edge connecting $u_x^{i}$ and $s_x$. Next, for every $i \in [2]$, every $C \in \mathcal{C}$ and every $x \in C$, we let $E(G)$ contain an edge connecting $u_x^{i}$ and $s_{x,C}^{i}$. Finally, for every $i \in [2]$, every $C \in \mathcal{C}$ and for every pair of vertices in $S_C^{i}$, we let $E(G)$ contain an edge connecting the vertices in this pair. This finishes the description of $(G,S,r)$; see Figure~\ref{fig:k2_hardness} for an illustration. It is easy to see that $(G,S,r)$ can be constructed in polynomial time from $(X,\mathcal{C})$, the size of $(G,S,r)$ is polynomial in the size of $(X,\mathcal{C})$, and $d_G(S)+|E(G[S])|=2|S|$. We claim that $(G,S,r)$ is a yes-instance of \twoori if and only if $(X,\mathcal{C})$ is a yes-instance of \textsc{MNAE-3-SAT}.

    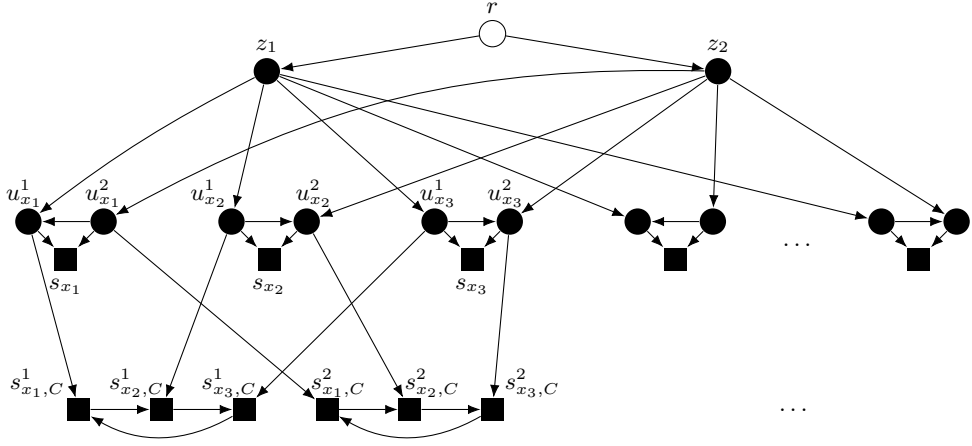
\begin{figure} \centering
        \begin{tikzpicture}
    		\tikzset{
    			every circle node/.style={minimum size=0.35cm, inner sep=0cm,font=\small},
    			every rectangle node/.style={minimum size=0.3cm, inner sep=0cm},font=\small,
    			>={Latex[scale=1]}
    		}
    		\def\d{2.7}
    		\def\e{1.1}
    		
    		\node[circle, draw, minimum size=1cm, inner sep=1cm, label={$r$}] (r) at (0,0) {};
    		\node[circle, fill, label={$z_1$}] (z1) at (-3,-0.5) {};
    		\node[circle, fill, label={$z_2$}] (z2) at (3,-0.5) {};
    		\draw[->] (r) -- (z1);
    		\draw[->] (r) -- (z2);
    		
    		\node[circle, fill, label={[xshift=-0.05cm]:$u^1_{x_1}$}] (u11) at ($(-2.1*\d-0.5,-2.5)$) {};
    		\node[circle, fill, label={$u^2_{x_1}$}] (u12) at ($(u11)+(1, 0)$) {};
    		\node[fill, label={[yshift=-0.05cm]270:$s_{x_1}$}] (s1) at ($(u11)+(0.5, -0.5)$) {};
    		\draw[->] (u11) -- (s1);
    		\draw[->] (u12) -- (s1);
    		\draw[->] (u12) -- (u11);
    		\path (z1) edge[->, bend right=5] (u11);
    		\path (z2) edge[->, bend right=15] (u12);
    		
    		\node[circle, fill, label={[xshift=-0.3cm]:$u^1_{x_2}$}] (u21) at ($(u11)+(\d, 0)$) {};
    		\node[circle, fill, label={[xshift=0.1cm]:$u^2_{x_2}$}] (u22) at ($(u12)+(\d, 0)$) {};
    		\node[fill,  label={[yshift=-0.05cm]270:$s_{x_2}$}] (s2) at ($(s1)+(\d, 0)$) {};
    		\draw[->] (u21) -- (s2);
    		\draw[->] (u22) -- (s2);
    		\draw[->] (u21) -- (u22);
    		\path (z1) edge[->] (u21);
    		\draw[->] (z2) -- (u22);
    		
    		\node[circle, fill, label={[xshift=0.05cm]:$u^1_{x_3}$}] (u31) at ($(u21)+(\d, 0)$) {};
    		\node[circle, fill, label={[xshift=-0.05cm]:$u^2_{x_3}$}] (u32) at ($(u22)+(\d, 0)$) {};
    		\node[fill,  label={[yshift=-0.05cm]270:$s_{x_3}$}] (s3) at ($(s2)+(\d, 0)$) {};
    		\draw[->] (u31) -- (s3);
    		\draw[->] (u32) -- (s3);
    		\draw[->] (u31) -- (u32);
    		\draw[->] (z1) -- (u31);
    		\draw[->] (z2) -- (u32);
    		
    		\node[circle, fill] (u41) at ($(u31)+(\d, 0)$) {};
    		\node[circle, fill] (u42) at ($(u32)+(\d, 0)$) {};
    		\node[rectangle, fill] (s4) at ($(s3)+(\d, 0)$) {};
    		\draw[->] (u41) -- (s4);
    		\draw[->] (u42) -- (s4);
    		\draw[->] (u42) -- (u41);
    		\draw[->] (z1) -- (u41);
    		\draw[->] (z2) -- (u42);

    		\node[circle, fill] (u51) at ($(u41)+(1.2*\d, 0)$) {};
    		\node[circle, fill] (u52) at ($(u42)+(1.2*\d, 0)$) {};
    		\node[fill=black] (s5) at ($(s4)+(1.2*\d, 0)$) {};
    		\draw[->] (u51) -- (s5);
    		\draw[->] (u52) -- (s5);
    		\draw[->] (u51) -- (u52);
    		\draw[->] (z1) -- (u51);
    		\draw[->] (z2) -- (u52);
    		
    		\node[label=270:{$\ldots$}] at ($(u42)!0.5!(u51)$) {};
    		
    		\node[fill, label={[xshift=-0.55cm]:$s^1_{x_1, C}$}] (s1C1) at (-5.5,-5) {};
    		\node[fill, label={[xshift=-0.33cm]90:{$s^1_{x_2, C}$}}] (s2C1) at ($(s1C1)+(\e, 0)$) {};
    		\node[fill, label={[xshift=-0.2cm]:$s^1_{x_3, C}$}] (s3C1) at ($(s2C1)+(\e, 0)$) {};
    		\draw[->] (s1C1) -- (s2C1);
    		\draw[->] (s2C1) -- (s3C1);
    		\path (s3C1) edge[bend left, ->] (s1C1);
    		\draw[->] (u11) -- (s1C1);
    		\draw[->] (u21) -- (s2C1);
    		\draw[->] (u31) -- (s3C1);
    		
    		\node[fill, label={[xshift=0.15cm]:$s^2_{x_1, C}$}] (s1C2) at ($(s3C1)+(\e, 0)$) {};
    		\node[fill, label={[xshift=0.3cm]:$s^2_{x_2, C}$}] (s2C2) at ($(s1C2)+(\e, 0)$) {};
    		\node[fill, label={[xshift=0.53cm]:$s^2_{x_3, C}$}] (s3C2) at ($(s2C2)+(\e, 0)$) {};
    		\draw[->] (s1C2) -- (s2C2);
    		\draw[->] (s2C2) -- (s3C2);
    		\path (s3C2) edge[bend left, ->] (s1C2);
    		\draw[->] (u12) -- (s1C2);
    		\draw[->] (u22) -- (s2C2);
    		\draw[->] (u32) -- (s3C2);
    		
    		\node at (4, -5) {$\ldots$};
	    \end{tikzpicture}
        \caption{An illustration of the proof of \cref{serdgh}. $C=\{x_1, x_2, x_3\}$ is a clause, and the orientation corresponds to a truth assignment $\phi$ with $\phi(x_1)=\false$ and $\phi(x_2)=\phi(x_3)=\true$.}         \label{fig:k2_hardness}
    \end{figure}
    
    First suppose that  $(X,\mathcal{C})$ is a yes-instance of \textsc{MNAE-3-SAT}, so there exists a satisfying assignment $\phi\colon X \to \{\true,\false\}$ for $(X,\mathcal{C})$. We now define an orientation $\dG$ of $G$. First, we orient all edges between $r$ and $Z$ from $r$ to $Z$, we orient all edges between $Z$ and $U$ from $Z$ to $U$, and we orient all edges between $U$ and $S$ from $U$ to $S$. Further, for every $C \in \mathcal{C}$ and $i \in [2]$, we orient the edges in $E(G[S_C^{i}])$ in a way that $\dG[S_C^{i}]$ is a directed triangle. Observe that this orientation is not unique, however this small ambiguity will be effectless. Finally, for every $x \in X$ with $\phi(x)=\true$, we orient the edge connecting $u_x^{1}$ and $u_x^{2}$ from $u_x^{1}$ to $u_x^{2}$ and for every $x \in X$ with $\phi(x)=\false$, we orient the edge connecting $u_x^{1}$ and $u_x^{2}$ from $u_x^{2}$ to $u_x^{1}$.
    
    It remains to show that $\lambda_{\dG}(r,s)\geq 2$ for all $s \in S$. 
    For $x \in X$, $rz_1u_x^1s_x$ and $rz_2u_x^2s_x$ are two arc-disjoint directed $r$-$s_x$ paths in $\dG$, so $\lambda_{\dG}(r,s_x)\geq 2$.
    Now consider $s_{x_1,C}^{1}$ for some $C \in \mathcal{C}$ and $x_1 \in C$.
    Let $x_2, x_3 \in X$ be the variables such that $C=\{x_1, x_2, x_3\}$ and $A(\dG[S^1_{C}]) = \{s^1_{x_1, C} s^1_{x_2, C}, s^1_{x_2, C} s^1_{x_3, C}, s^1_{x_3, C} s^1_{x_1, C}\}$.
    As $\phi$ is satisfying assignment, at least one of $\phi(x_1)$, $\phi(x_2)$, and $\phi(x_3)$ is $\false$.
    In each case, $\dG$ contains two arc-disjoint directed $r$-$s^1_{x_1, C}$ paths: $rz_1u^1_{x_3}s^1_{x_3,C}s^1_{x_1, C}$ and $rz_2u^2_{x_1}u^1_{x_1}s^1_{x_1,C}$ if $\phi(x_1)=\false$, $rz_1u^1_{x_1}s^1_{x_1,C}$ and $rz_2u^2_{x_2}u^1_{x_2}s^1_{x_2,C}s^1_{x_3,C}s^1_{x_1,C}$ if $\phi(x_2)=\false$, 
    and $rz_1u^1_{x_1}s^1_{x_1,C}$ and $rz_2u^2_{x_3}u^1_{x_3}s^1_{x_3,C}s^1_{x_1,C}$ if $\phi(x_3)=\false$.
    This shows that $\lambda_{\dG}(r,s_{x_1,C}^1)\geq 2$. A similar argument shows that  $\lambda_{\dG}(r,s_{x,C}^2)\geq 2$ for all $C \in \mathcal{C}$ and $x \in C$.
    
    Now suppose that $(G,S,r)$ is a yes-instance of \twoori, so there exists a Steiner rooted $2$-orientation $\dG$ of $G$.
    First consider some arbitrary $s \in S$.
    Observe that, as $\dG$ is Steiner rooted $2$-arc-connected and by construction, we have that $2|S|\leq \sum_{s \in S}\lambda_{\dG}(r,s)\leq\sum_{s \in S} d_{\dG}^{in}(s) \leq d_G(S)+|E(G[S])|=2|S|$. Hence equality holds throughout and so the edges between $U$ and $S$ are oriented from $U$ to $S$ in $\dG$.
    Next, consider some $x \in X$ and let $U_x\coloneqq \{u_x^1,u_x^2,s_x\}$.
    As all edges between $U$ and $S$ are oriented from $U$ to $S$ in $\dG$ and $\dG$ is Steiner rooted $2$-arc-connected, we obtain that $2\leq \lambda_{\dG}(r,s_x)\leq d_{\dG}^{in}(U_x)\leq d_G(Z,U_x)=2$.
    It follows that equality holds throughout and all edges between $Z$ and $U$ are oriented from $Z$ to $U$ in $\dG$.
    
    We now define a function $\phi\colon X \to \{\true,\false\}$ as follows: for every $x \in X$, we set $\phi(x)=\true$ if the edge connecting $u_x^1$ and $u_x^2$ is oriented from $u_x^1$ to $u_x^2$ and we set $\phi(x)=\false$ otherwise. %if the edge connecting $u_x^1$ and $u_x^2$ is oriented from $u_x^2$ to $u_x^1$.
    Suppose for the sake of a contradiction that there exists some clause $C\in \mathcal{C}$ containing the variables $x_1$, $x_2$, and $x_3$ such that $\phi(x_1)=\phi(x_2)=\phi(x_3)=\true$. Let $Y\coloneqq S^1_C \cup \{z_1,u_{x_1}^1,u_{x_2}^1,u_{x_3}^1\}$.
    As all edges between $Z$ and $U$ are oriented from $Z$ to $U$ in $\dG$, all edges between $U$ and $S$ are oriented from $U$ to $S$ in $\dG$, and by assumption and construction, we have that $d_{\dG}^{in}(Y)\leq 1$. We obtain that $\lambda_{\dG}(r,s_{x_1,C}^1)\leq d_{\dG}^{in}(Y)\leq 1$, a contradiction to $\dG$ being Steiner rooted $2$-arc-connected. A similar argument shows that  $\phi(x_1)=\phi(x_2)=\phi(x_3)=\false$ is impossible. It follows that $\phi$ is a satisfying assignment, so $(X,\mathcal{C})$ is a yes-instance of \textsc{MNAE-3-SAT}.
    
    The statement now follows from the fact that \textsc{MNAE-3-SAT} is NP-hard by \cref{vfcds}.
\end{proof}

We are now ready to prove \cref{kfix} which we restate here for convenience.

\kfix*

\begin{proof}
    Let $k \geq 2$ be fixed. By \cref{serdgh}, we may suppose that $k \geq 3$. We prove the statement by a reduction based on \cref{serdgh}. Let $(G,S,r)$ be an instance of \twoori for which $d_G(S)+|E(G[S])|=2|S|$ holds. We now create a graph $G'$ from $G$ by adding a set $W$ of $k-2$ new vertices and adding a new edge connecting $w$ and $s$ for every $w \in W$ and every $s \in S+r$. Clearly, $G'$ can be computed from $G$ in polynomial time and the size of $G'$ is polynomial in the size of $G$. We claim that $(G',S,r)$ is a yes-instance of \kori if and only if $(G,S,r)$ is a yes-instance of \twoori.
    
    First suppose that $(G,S,r)$ is a yes-instance of \twoori, so there exists an orientation $\dG$ of $G$ such that $\dG$ is Steiner rooted $2$-arc-connected. We now define an orientation $\dGp$ of $G'$ by giving every edge in $E(G)$ the orientation it has in $\dG$, orienting all edges between $r$ and $W$ from $r$ to $W$ and orienting all edges between $W$ and $S$ from $W$ to $S$. Now consider some $s \in S$. As $\dG$ is Steiner rooted $2$-arc-connected, there exist two arc-disjoint directed  $r$-$s$ paths $P_1, P_2$ in $\dG$ and hence in $\dGp[V(G)]$. Further, for every $w\in W$, observe that $rws$ is a directed $r$-$s$ path in $\dGp$. As all of these directed paths are arc-disjoint, it follows that $\lambda_{\dGp}(r,s)\geq k$, so  $\dGp$ is Steiner rooted $k$-arc-connected. It follows that $(G',S,r)$ is a yes-instance of  \kori.
    
    Now suppose that $(G',S,r)$ is a yes-instance of  \kori, so there exists an orientation $\dGp$ of $G'$ such that $\dGp$ is Steiner rooted $k$-arc-connected. Let $\dG\coloneqq \dGp[V(G)]$. By assumption and construction, we have that $d_{G'}(S)+|E(G'[S])|=k|S|$ holds. As $\dGp$ is Steiner rooted $k$-arc-connected, we obtain that $k|S|\leq \sum_{s \in S}\lambda_{\dGp}(r,s)\leq \sum_{s \in S}d_{\dGp}^{in}(s)\leq d_{G'}(S)+|E(G'[S])|=k|S|$, so equality holds throughout. In particular, it follows that all edges between $W$ and $V(G)-r$ are oriented from $W$ to $V(G)-r$, so any directed path containing an arc in $A(\dG)$ is fully contained in $\dGp[V(G)]$. Now consider some $s \in S$. As $\dGp$ is Steiner rooted $k$-arc-connected, there exist $k$ arc-disjoint paths from $r$ to $S$ in $\dGp$. As $d_{G'}(r,W)=k-2$, it follows that at least two of these directed paths contain an arc in $A(\dG)$. By the above remark, it follows that these two directed paths are fully contained in $\dG$. This yields that  $\dG$ is Steiner rooted $2$-arc-connected, so $(G,S,r)$ is a yes-instance of \twoori.

    The theorem now follows from \cref{serdgh}.
\end{proof}

\begin{rem}
    Frank, Kir\'aly, and Kir\'aly~ \cite{frank2003orientation} proved that {\sc $R$-orientation} is NP-hard even if $R(u,v)\leq 3$ for all $(u,v)\in V\times V$. \Cref{kfix} improves this bound from $3$ to $2$. 
\end{rem}

%%%%%%%%%%%%%%%%%%%%%%%%%%%%%%%%
\subsection{Hardness for Fixed \texorpdfstring{$t$}{t} with \texorpdfstring{$k$}{k} as Part of the Input}
\label{sec:fixt}
%%%%%%%%%%%%%%%%%%%%%%%%%%%%%%%%

In this section we prove \cref{tfix}. 
\Cref{3colsat} introduces a new version of the satisfiability problem that will be used in our reduction. \Cref{dircut} shows that the problem is hard even with three terminals, assuming that a certain set of edges is preoriented. Finally, \Cref{concludetfix} explains how to adapt this result to prove \Cref{tfix}.

\subsubsection{A Modified SAT Problem} \label{3colsat}

An instance of \textsc{Maximum 2-Satisfiability} (\textsc{MAX-2-SAT}) consists of a set $X$ of binary variables, a set $\mathcal{C}$ of clauses over $X$ such that each clause contains exactly 2 literals over $X$, and an integer $k$.
We say that a clause $C \in \mathcal{C}$ {\it contains} a variable $x$ if it contains one of the literals $x$ and $\bar{x}$.
For an assignment $\phi\colon X \rightarrow \{\true, \false\}$, we say that some $C \in \mathcal{C}$ is satisfied if at least one of the two literals contained in $C$ is set to $\true$.
The problem is to decide whether there exists an assignment for $(X,\mathcal{C})$ satisfying at least $k$ clauses of $\mathcal{C}$. We will use the following well-known hardness result~\cite{garey1976some}.

\begin{thm}[Garey, Johnson, Stockmeyer]\label{thm:max2sat}
    \textsc{MAX-2-SAT} is NP-hard.
\end{thm}

Given an instance $(X,\mathcal{C})$ of \textsc{MAX-2-SAT}, we define its {\it clause intersection graph} to be the graph $H$ with $V(H)=\mathcal{C}$ and in which two clauses $C_1$ and $C_2$ are linked by an edge if there exists a variable that is contained in both of them.
We define the \textsc{3-Colored Maximum 2-Satisfiability} (\textsc{3-COL-MAX-2-SAT}) problem as follows. 
An instance of the problem consists of a set $X$ of binary variables, a set $\mathcal{C}$ of clauses each containing exactly two literals over $X$, a proper coloring $\psi\colon  V(H)\rightarrow [3]$ of the clause intersection graph $H$ of $(X,\mathcal{C})$ and 3 integers $k_1,k_2,k_3$. The question is whether there exists an assignment $\phi\colon X \rightarrow \{\true, \false\}$ such that for every $i \in [3]$, at least $k_i$ of the clauses in $\psi^{-1}(i)$ are satisfied. We prove the following hardness result.

\begin{lem}\label{3colhard} 
    \textsc{3-COL-MAX-2-SAT} is NP-hard.
\end{lem}
\begin{proof}
    We prove this statement by a reduction from \textsc{MAX-2-SAT}.
    Let $(X,\cC, k)$ be an instance of \textsc{MAX-2-SAT} where $X=\{x_1,\ldots, x_n\}$ and let $m\coloneqq |\cC|$.
    We may assume that each variable is contained in each clause at most once.
    For $i \in [n]$, let $m_i$ denote the number of clauses containing $x_i$, and define $m'_i \coloneqq m_i$ if $m_i$ is even and $m'_i\coloneq m_i+1$ otherwise.
    Let $x_{i,j}$ be a new binary variable for each $i \in [n]$ and $j \in [m'_i]$, and let $X'\coloneqq\{x_{i,j}: i \in [n], j \in [m'_i]\}$.
    We define an instance of \textsc{3-COL-MAX-2-SAT} with set of variables $X'$.
    Consider any ordering of the clauses in $\cC$, and for each clause $C \in \cC$, replace each variable $x_i$ contained in $C$ by $x_{i,j}$ if $C$ is the $j$th clause containing $x_i$ (while keeping the possible negation of the variable).
    Let $\tilde{\cC}$ denote the thus obtained clauses.
    For $i \in [n]$, let
    \[\cC_i\coloneqq \{(x_{i,1}\vee \widebar{x_{i,2}}), ~ (x_{i,2}\vee \widebar{x_{i,3}}), \ldots, ~ (x_{i,m'_i-1}\vee \widebar{x_{i,m'_i}}), ~ (x_{i,m'_i}\vee \widebar{x_{i,1}})\}.\]
    Let $\cC'\coloneqq \tilde{\cC} \cup \cC_1\cup \ldots \cup \cC_n$.
    Observe that any two clauses in $\tilde{\cC}$ contain pairwise distinct variables, and so do any clause from $\cC_{i_1}$ and any clause from $\cC_{i_2}$ for $i_1, i_2 \in [n]$ with $i_1 \ne i_2$. 
    For $i \in [n]$, as $m'_i$ is even, the clause intersection graph of $(X', \cC_i)$ is a cycle of even length.
    This shows that a proper 3-coloring $\psi$ of the intersection graph of $(X',\cC')$ can be computed in polynomial time: each $\cC_i$ can be properly colored using the colors 1 and 2, and by coloring the clauses in $\tilde{C}$ to 3, we get a proper 3-coloring.
    Finally, we set $k_1\coloneqq |\psi^{-1}(1)|$, $k_2 \coloneqq |\psi^{-1}(2)|$, and $k_3\coloneqq k$. 
    It is clear that $(X',\cC',\psi, k_1, k_2, k_3)$ is an instance of \textsc{3-COL-MAX-2-SAT} which can be computed in polynomial time.
    We claim that $(X,\cC, k)$ is a yes-instance of \textsc{MAX-2-SAT} if and only if $(X', \cC', \psi, k_1, k_2, k_3)$ is a yes-instance of \textsc{3-COL-MAX-2-SAT}.

    Assume first that $(X, \cC, k)$ is a yes-instance, that is, there is an assignment $\phi\colon X \to \{\true, \false\}$ satisfying at least $k$ clauses of $\cC$.
    Define the assignment $\phi'\colon X'\to \{\true, \false\}$ such that $\phi'(x_{i,1})=\ldots = \phi'(x_{i,m'_i}) = \phi(x_i)$ for each $i \in [n]$.
    Then each clause in $\cC_i$ is satisfied for each $i \in [n]$, and at least $k$ clauses in $\tilde{C}$ are also satisfied.
    This shows that $\phi'$ satisfies at least $k_t$ clauses of $\psi^{-1}(t)$ for $t \in [3]$, thus $(X', \cC', \psi, k_1, k_2, k_3)$ is a yes-instance of 3-COL-MAX-2-SAT.
    
    Conversely, suppose that $(X', \cC', \psi, k_1, k_2, k_3)$ is a yes-instance, that is, there is an assignment $\phi'\colon X \to \{\true, \false\}$ satisfying at least $k_t$ clauses of $\psi^{-1}(t)$ for $t \in [3]$.
    Since $k_1=|\psi^{-1}(1)|$, $k_2=|\psi^{-1}(2)|$, and $\cC_i$ is colored using the colors 1 and 2, $\phi'$ satisfies each clause of $\cC_i$ for $i \in [n]$.
    This implies that $\phi'(x_{i,1})=\ldots = \phi'(x_{i,m'_i})$ holds for $i \in [n]$, thus we can define $\phi\colon X \to \{\true, \false\}$ such that $\phi(x_i)=\phi'(x_{i,1})=\ldots=\phi'(x_{i,m'_i})$ holds for $i \in [n]$.
    It is clear that the number of clauses of $\cC$ satisfied by $\phi$ is the same as the number of clauses of $\tilde{C}$ satisfied by $\phi'$, thus $(X, \cC, k)$ is a yes-instance of \textsc{MAX-2-SAT}.  This finishes the proof using \cref{thm:max2sat}.  
\end{proof}

%%%%%%%%%%%%%%%%%%%%%%%%%%%%%%%%
\subsubsection{Hardness Under Further Constraints}
\label{dircut}
%%%%%%%%%%%%%%%%%%%%%%%%%%%%%%%%

We here deal with a modified version of \tori that serves as an intermediate step in our reduction. Namely, an instance of \mtori consists of a graph $G$, a vertex $r \in V(G)$, a set $S \subseteq V(G)-r$, an integer $k$, and a set $Y \subseteq V(G)\setminus S$ with $r \in Y$. The question is whether there exists an orientation $\dG$ of $G$ such that $\lambda_{\dG}(r,s)\geq k$ holds for all $s \in S$ and all edges in $\delta_G(Y)$ are oriented away from $Y$. We prove the following result.

\begin{lem}\label{etxdrcdzfvzugbu}
    \mthreeori is NP-hard.
\end{lem}
\begin{proof}
    We prove this statement by a reduction from \textsc{3-COL-MAX-2-SAT}. Let $(X,\mathcal{C},\psi,k_1,k_2,k_3)$ be an instance of \textsc{3-COL-MAX-2-SAT}. For $i \in [3]$, we let $\mathcal{C}_i$ consist of all $C \in \mathcal{C}$ with $\psi(C)=i$. 
    We may assume that $k_i \le |\cC_i|$ for $i \in [3]$.
    We now create an instance $(G,S,r,k,Y)$ of \mthreeori as follows: first, we let $V(G)$ consist of two vertices $u_x$ and $u_{\bar{x}}$ for every $x \in X$, a vertex $v_C$ for every $C \in \mathcal{C}$, a vertex $r$ and a set of vertices $S=\{s_1,s_2,s_3\}$. We first let $E(G)$ contain an edge connecting $r$ and $u_\ell$ for every literal $\ell$ over $X$ and an edge connecting $u_x$ and $u_{\bar{x}}$ for every $x \in X$. Next, for every $C \in \mathcal{C}$ and every literal $\ell$ contained in $C$, we let $E(G)$ contain two parallel edges connecting $u_\ell$ and $v_C$. Further, for every $i \in [3]$ and every $C \in \mathcal{C}_i$, we add 3 parallel edges connecting $v_C$ and $s_i$. We now set $k\coloneqq 3|\mathcal{C}|$ and for every $i \in [3]$, we add $k-(2|\mathcal{C}_i|+k_i)$ parallel edges connecting $r$ and $s_i$. Finally, we let $Y$ consist of $r$ and the vertex $u_\ell$ for every literal $\ell$ over $X$. This finishes the description of $(G,S,r,k,Y)$; see Figure~\ref{fig:mod_Steiner} for an illustration. It is easy to see that $(G,S,r,k,Y)$ can be constructed from $(X,\mathcal{C},\psi,k_1,k_2,k_3)$ in polynomial time. We claim that $(G,S,r,k,Y)$ is a yes-instance of \mtori if and only if $(X,\mathcal{C},\psi,k_1,k_2,k_3)$ is a yes-instance of \textsc{3-COL-MAX-2-SAT}.

    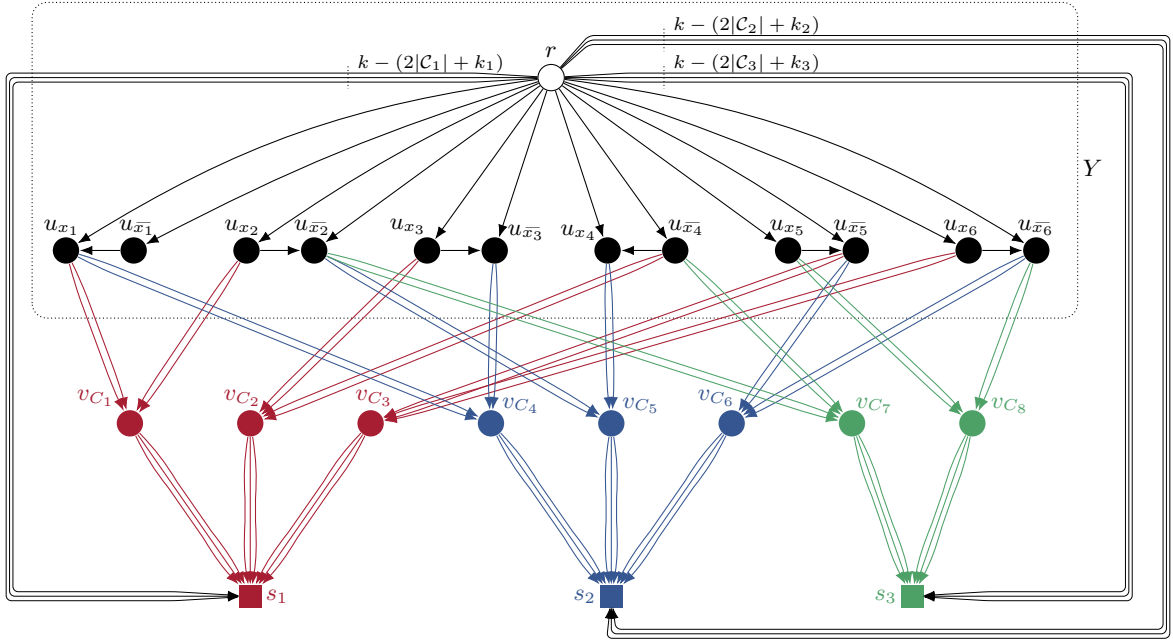
\begin{figure}
        \centering
    	\newcommand{\threeparallel}[3]{%
    		\draw[color=#3, ->] (#1) to[bend left=6] ($(#1)!1cm!4:(#2)$) to ($(#2)!1cm!-4:(#1)$) to[bend left=6] (#2);
    		\draw[color=#3,->] (#1) to[bend right=6] ($(#1)!1cm!-4:(#2)$) to ($(#2)!1cm!4:(#1)$) to[bend right=6] (#2);
    		\draw[color=#3, ->] (#1) to (#2);
    	}
    	\newcommand{\twoparallel}[3]{
    		\draw[color=#3, ->] (#1) to[bend left=2] ($(#1)!0.7cm!4:(#2)$) to ($(#2)!0.7cm!-4:(#1)$) to[bend left=2] (#2);
    		\draw[color=#3, ->] (#1) to[bend right=2] ($(#1)!0.7cm!-4:(#2)$) to ($(#2)!0.7cm!4:(#1)$) to[bend right=2] (#2);
    	}
    	\newcommand{\twotwoparallel}[4]{
    		\twoparallel{#1}{#3}{#4};
    		\twoparallel{#2}{#3}{#4};
    	}
    	
    	\begin{tikzpicture}
    		\tikzset{
    			every circle node/.style={minimum size=0.35cm, inner sep=0cm},
    			every rectangle node/.style={minimum size=0.3cm, inner sep=0cm},
    			font=\small,
    			>={Latex[scale=1]}
    		}
    		\def\d{2.4}
    		\def\e{1.6}
    		\def\h{2.3}
    
    		\node[circle, draw, minimum size=1cm, inner sep=1cm, label={$r$}] (r) at (0,0) {};
    		\node[circle, fill, label={[xshift=-0.05cm,yshift=-0.03cm]:$u_{x_1}$}] (u11) at ($(-2.5*\d-0.45,-\h)$) {};
    		\node[circle, fill, label={[xshift=0.05cm,yshift=-0.03cm]:$u_{\widebar{x_1}}$}] (u12) at ($(u11)+(0.9, 0)$) {};
    		\node[circle, fill, label={[xshift=-0.05cm,yshift=-0.03cm]:$u_{x_2}$}] (u21) at ($(u11)+(\d, 0)$) {};
    		\node[circle, fill, label={[xshift=-0.02cm,yshift=-0.03cm]:$u_{\widebar{x_2}}$}] (u22) at ($(u12)+(\d, 0)$) {};
    		\node[circle, fill, label={[xshift=-0.25cm,yshift=-0.02cm]:$u_{x_3}$}] (u31) at ($(u21)+(\d, 0)$) {};
    		\node[circle, fill, label={[xshift=0.42cm, yshift=-0.1cm]:$u_{\widebar{x_3}}$}] (u32) at ($(u22)+(\d, 0)$) {};
    		\node[circle, fill, label={[xshift=-0.38cm, yshift=-0.1cm]:$u_{x_4}$}] (u41) at ($(u31)+(\d, 0)$) {};
    		\node[circle, fill, label={[xshift=0.15cm,yshift=-0.03cm]:$u_{\widebar{x_4}}$}] (u42) at ($(u32)+(\d, 0)$) {};
    		\node[circle, fill, label={[xshift=0.00cm,yshift=-0.03cm]:$u_{x_5}$}] (u51) at ($(u41)+(\d, 0)$) {};
    		\node[circle, fill, label={[xshift=-0.05cm,yshift=-0.03cm]:$u_{\widebar{x_5}}$}] (u52) at ($(u42)+(\d, 0)$) {};
    		\node[circle, fill, label={[xshift=-0.09cm,yshift=-0.03cm]:$u_{x_6}$}] (u61) at ($(u51)+(\d, 0)$) {};
    		\node[circle, fill, label={[xshift=0.00cm,yshift=-0.03cm]:$u_{\widebar{x_6}}$}] (u62) at ($(u52)+(\d, 0)$) {};
    
    		\draw[<-] (u11) -- (u12);
    		\draw[->] (u21) -- (u22);		
    		\draw[->] (u31) -- (u32);
    		\draw[<-] (u41) -- (u42);
    		\draw[->] (u51) -- (u52);
    		\draw[->] (u61) -- (u62);
    		
    		\draw[bend right=12, ->] (r) to (u11);
    		\draw[bend right=5, ->] (r) to (u12);
    		\draw[bend right=5, ->] (r) to (u21);
    		\draw[->] (r) to (u22);
    		\draw[->] (r) to (u31);
    		\draw[->] (r) to (u32);
    		\draw[->] (r) to (u41);
    		\draw[->] (r) to (u42);
    		\draw[->] (r) to (u51);
    		\draw[bend left=5, ->] (r) to (u52);
    		\draw[bend left=5, ->] (r) to (u61);
    		\draw[bend left=12, ->] (r) to (u62);
    		
    		\node[circle, myred, fill,label={[myred, xshift=-0.45cm]:$v_{C_1}$}] (v1) at ($(-3.5*\e,-2*\h)$) {};
    		\node[circle, myred, fill, label={[myred, xshift=-0.1cm]:$v_{C_2}$}] (v2) at ($(v1)+(\e, 0)$) {};
    		\node[circle, fill, myred, label={[myred, xshift=0.05cm]:$v_{C_3}$}] (v3) at ($(v2)+(\e, 0)$) {};
    		\node[circle, fill, myblue, label={[myblue, xshift=0.4cm, yshift=-0.1cm]:$v_{C_4}$}] (v4) at ($(v3)+(\e, 0)$) {};
    		\node[circle, fill, myblue,  label={[myblue, xshift=0.4cm, yshift=-0.1cm]:$v_{C_5}$}] (v5) at ($(v4)+(\e, 0)$) {};
    		\node[circle, fill, myblue, label={[myblue, xshift=-0.2cm]:$v_{C_6}$}] (v6) at ($(v5)+(\e, 0)$) {};
    		\node[circle, fill, mygreen, label={[mygreen, xshift=0.3cm, yshift=-0.1cm]:$v_{C_7}$}] (v7) at ($(v6)+(\e, 0)$) {};
    		\node[circle, fill, mygreen, label={[mygreen, xshift=0.5cm, yshift=-0.1cm]:$v_{C_8}$}] (v8) at ($(v7)+(\e, 0)$) {};
    		
    		\draw[rounded corners=2mm, densely dotted, draw] (-6.9,-3.2) rectangle (7,1);
    		\node at (7.2, -1.2) {$Y$};
    		
    		\twotwoparallel{u11}{u21}{v1}{myred};
    		\twotwoparallel{u31}{u42}{v2}{myred};
    		\twotwoparallel{u52}{u61}{v3}{myred};
    		
    		\twotwoparallel{u11}{u32}{v4}{myblue};
    		\twotwoparallel{u22}{u41}{v5}{myblue};
    		\twotwoparallel{u52}{u62}{v6}{myblue};
    
    		\twotwoparallel{u22}{u42}{v7}{mygreen};
    		\twotwoparallel{u62}{u51}{v8}{mygreen};
    		
    		\node[fill, myred, label={[myred, xshift=0.04cm]0:$s_1$}] (s1) at ($(v2)+(0, -\h)$) {};
		      \node[fill, myblue, label={[myblue, xshift=-0.04cm]180:$s_2$}] (s2) at ($(v5)+(0, -\h)$) {};
		      \node[fill, mygreen, label={[mygreen, xshift=-0.04cm]180:$s_3$}] (s3) at ($(v7)!0.5!(v8) + (0,-\h)$) {};
	
    		\threeparallel{v1}{s1}{myred};
    		\threeparallel{v2}{s1}{myred};
    		\threeparallel{v3}{s1}{myred};
    		\threeparallel{v4}{s2}{myblue};
    		\threeparallel{v5}{s2}{myblue};
    		\threeparallel{v6}{s2}{myblue};
    		\threeparallel{v7}{s3}{mygreen};
    		\threeparallel{v8}{s3}{mygreen};
            \foreach \i in {-1.5,0,1.5} {
    			\draw[rounded corners=2, ->, very thin] (r) -- ($(r)+(-1, \i*0.04)$) -- ($(-3*\d-\i*0.03, \i*0.04)$) -- ($(-3*\d-0.03*\i, -3*\h-\i*0.04)$) -- ($(s1)+(-1,-\i*0.04)$) -- (s1);
    		}
    		\foreach \i in {-1.5,0,1.5} {
    			\draw[rounded corners=2, ->, very thin] (r) -- ($(r)+(1, \i*0.04)$) -- ($(3.2*\d+\i*0.03, \i*0.04)$) -- ($(3.2*\d+0.03*\i, -3*\h-\i*0.04)$) -- ($(s3)+(1,-\i*0.04)$) -- (s3);
    		}
    		\foreach \i in {-1.5,0,1.5} {
    			\draw[rounded corners=2, ->, very thin] (r) -- ($(r)+(0.5, 0.5+\i*0.04)$) -- ($(3.2*\d+\i*0.03+0.5, 0.5+\i*0.04)$) -- ($(3.2*\d+0.03*\i+0.5, -3*\h-\i*0.04-0.5)$) -- ($(s2)+(-\i*0.04,-0.5-\i*0.04)$) -- (s2);
    		}
            \draw[densely dotted] (-2.7,0.15) -- (-2.7, -0.15);  
            \node[label={\scriptsize $k-(2|\cC_1|+k_1)$}] at (-1.6, -0.12) {};
            \draw[densely dotted] (1.5,0.15) -- (1.5, -0.15);  
             \node[label={\scriptsize $k-(2|\cC_3|+k_3)$}] at (2.6, -0.12) {};
             \draw[densely dotted] (1.5,0.65) -- (1.5, 0.35);  
             \node[label={\scriptsize $k-(2|\cC_2|+k_2)$}] at (2.6, 0.38) {};
    	\end{tikzpicture}
        \caption{An illustration of the proof of \cref{etxdrcdzfvzugbu}. The set of clauses consists of $\cC_1=\{(x_1\vee x_2), (x_3\vee \widebar{x_4}), (\widebar{x_5}\vee x_6)\}$, $\cC_2=\{(x_1\vee \widebar{x_3}, ), (\widebar{x_2}\vee x_4), (\widebar{x_5}\vee \widebar{x_6}))\}$, and $\cC_3=\{(\widebar{x_2}\vee \widebar{x_4}), (x_5 \vee \widebar{x_6})\}$. The orientation corresponds to the truth assignment $\phi$ with $\phi(x_1)=\phi(x_4)=\true$, $\phi(x_2)=\phi(x_3)=\phi(x_5)=\phi(x_6)=\false$.}
        \label{fig:mod_Steiner}
    \end{figure}

    First suppose that $(X,\mathcal{C},\psi,k_1,k_2,k_3)$ is a yes-instance of \textsc{3-COL-MAX-2-SAT}, so there exists an assignment $\phi\colon X\rightarrow \{\true, \false\}$ such that for $i \in [3]$, at least $k_i$ of the clauses in $\mathcal{C}_i$ are satisfied. We now define an orientation $\dG$ of $G$ as follows: first, we orient all edges in $\delta_G(r)$ away from $r$ and we orient all edges in $\delta_G(S)$ toward $S$.
    Next, for every $C \in \mathcal{C}$ and every literal $\ell$ contained in $C$, we orient the edge connecting $u_\ell$ and $v_C$ from $u_\ell$ to $v_C$. Finally, for every $x \in X$ with $\phi(x)=\true$, we orient the edge connecting $u_x$ and $u_{\bar{x}}$ from $u_{\bar{x}}$ to $u_x$ and for every $x \in X$ with $\phi(x)=\false$, we orient the edge connecting $u_x$ and $u_{\bar{x}}$ from $u_x$ to $u_{\bar{x}}$. This finishes the description of $\dG$.
    Clearly, we have that all the edges in $\delta_G(Y)$ are oriented away from $Y$ in $\dG$. It hence suffices to prove that $\lambda_{\dG}(r,s_i)\geq k$ holds for every $i \in [3]$. 
    To this end, fix some $i \in [3]$ and let $\mathcal{C}'_i$ be the set of clauses in $\mathcal{C}_i$ that are satisfied by $\phi$.
    Further, for every $C \in \mathcal{C}_i$ that contains some literals $\ell_1$ and $\ell_2$, let $A_C\subseteq A(\dG)$ be the set of arcs that are orientations of edges in $E(G)$ incident to at least one of the vertices $u_{\ell_1}$ ,$u_{\widebar{\ell_1}}$     $u_{\ell_2}$, $u_{\widebar{\ell_2}}$, and $v_C$.
    Observe that $A_C$ contains arc-disjoint directed $r$-$s_i$ paths $ru_{\ell_1}v_Cs_i$ and $ru_{\ell_2}v_Cs_i$, and if $C\in \cC_i$, then the edge $u_{\widebar{\ell_j}}u_{\ell_j}$ is oriented towards $u_{\ell_j}$ for some $j\in [2]$, thus there is a third arc-disjoint directed $r$-$s_i$ path $ru_{\widebar{\ell_j}} u_{\ell_j} v_C s_i$.
    As $\psi$ is a proper coloring, $A_{C_1}$ and $A_{C_2}$ are disjoint for all distinct $C_1, C_2 \in \cC_i$.
    This shows that there exist $2|\cC_i|+|\cC'_i|$ arc-disjoint directed $r$-$s_i$ paths contained in $\bigcup_{C\in \cC_i} A_C$.
    Moreover, $\delta_{\dG}^{out}(Z)$ contains $k-(2|\mathcal{C}_i|+k_i)$ arcs from $r$ to $s_i$ by construction, and we have $|\mathcal{C}'_i|\geq k_i$ by the choice of $\phi$.
    These together show that $\lambda_{\dG}(r,s_i) \ge 2|\cC_i|+|\cC'_i|+k-(2|\mathcal{C}_i|+k_i) \ge k$.
    It follows that $(G,S,r,k,Y)$ is a yes-instance of \mtori.

    Now suppose that $(G,S,r,k,Y)$ is a yes-instance of \mtori, so there exists an orientation $\dG$ of $G$ such that all edges in $\delta_G(Y)$ are oriented away from $Y$ and $\lambda_{\dG}(r,s_i)\geq k$ holds for all $i \in [3]$. We now define an assignment $\phi\colon X \to\{\true, \false\}$ as follows: for every $x \in X$ for which the edge connecting $u_x$ and $u_{\bar{x}}$ is oriented from $u_{\bar{x}}$ to $u_x$, we set $\phi(x)=\true$ and for every $x \in X$ for which the edge connecting $u_x$ and $u_{\bar{x}}$ is oriented from $u_{x}$ to $u_{\bar{x}}$, we set $\phi(x)=\false$. We still need to show that for every $i \in [3]$, at least $k_i$ of the clauses in $\mathcal{C}_i$ are satisfied by $\phi$. To this end, fix some $i \in [3]$, let $\mathcal{C}'_i$ be the set of clauses in $\mathcal{C}_i$ that are satisfied by $\phi$ and let $\mathcal{C}''_i\coloneqq\mathcal{C}_i\setminus \mathcal{C}'_i$. Now let $Z \subseteq V(G)$ be the set that contains $S \setminus \{s_i\}$, the vertex $v_C$ for all $C \in \mathcal{C}\setminus \mathcal{C}''_i$, the vertices $u_x$ and $u_{\bar{x}}$ for all variables $x \in X$ that are not contained in a clause of $\mathcal{C}''_i$ and, for all $C=\{\ell_1,\ell_2\}\in \mathcal{C}''_i$, the vertices $u_{\widebar{\ell_1}}$ and $u_{\widebar{\ell_2}}$. For every $a \in \delta_{\dG}^{in}(s_i)\cap \delta_{\dG}^{out}(Z)$, we have by construction that $a$ is either an orientation of one of the $k-(2|\mathcal{C}_i|+k_i)$ edges connecting $r$ and $s_i$ or an orientation of one of the 3 edges connecting $v_C$ and $s_i$ for some $C \in \mathcal{C}_i$. Now consider some $C\in \mathcal{C}''_i$. By the definition of $Z$, we have $\delta_{\dG}^{out}(Z)\cap \delta_{\dG}^{in}(v_C)=\emptyset$. Now let $\ell$ be a literal contained in $C$. As $C$ is not satisfied by $\phi$, as all edges in $\delta_G(Y)$ are oriented away from $Y$ in $\dG$ and by construction, we obtain that $|\delta_{\dG}^{out}(Z)\cap \delta_{\dG}^{in}(u_\ell)|\leq 1$. We obtain that $d_{\dG}^{out}(Z)\leq k-(2|\mathcal{C}_i|+k_i)+3|\mathcal{C}'_i|+2|\mathcal{C}''_i|=k+|\mathcal{C}'_i|-k_i$. On the other hand, as $r \in Z, s_i \in V(G)\setminus Z$ and $\lambda_{\dG}(r,s_i)\geq k$, we have $d_{\dG}^{out}(Z)\geq k$. We obtain $|\mathcal{C}'_i|\geq k_i$. Hence $(X,\mathcal{C},\psi,k_1,k_2,k_3)$ is a yes-instance of \textsc{3-COL-MAX-2-SAT}.
\end{proof}

%%%%%%%%%%%%%%%%%%%%%%%%%%%%%%%%
\subsubsection{Proof of \texorpdfstring{\cref{tfix}}{Theorem 1.3}}
\label{concludetfix}
%%%%%%%%%%%%%%%%%%%%%%%%%%%%%%%%

In this section we derive \cref{tfix} from \cref{etxdrcdzfvzugbu}. 
We will use the following simple observation on connectivities in directed graphs.

\begin{lem}\label{vuzuighj}
    Let $D$ be a directed graph and $r,v,s\in V(D)$. Further, let $D'$ be obtained from $D$ by adding an arc from $r$ to $v$ and an arc from $v$ to $s$. Then $\lambda_{D'}(r,s)=\lambda_D(r,s)+1$.
\end{lem}
\begin{proof}
    Let $Z \subseteq V(D)$ with $r \in Z$ and $s \in V(D)\setminus Z$. If $v \in Z$, then $vs$ is the unique arc in $\delta_{D'}^{out}(Z)\setminus \delta_{D}^{out}(Z)$. If $v \in V(D)\setminus Z$, then $rv$ is the unique arc in $\delta_{D'}^{out}(Z)\setminus \delta_{D}^{out}(Z)$. It follows that $d_{D'}^{out}(Z)=d_D^{out}(Z)+1$, from which the statement follows directly.
\end{proof}

We are now ready to give the main proof of \cref{tfix} which we restate here for convenience.

\tfix*

\begin{proof}
Given an instance $(G,S,r,k)$ of \fourori and some $t \geq 4$, we obtain an equivalent instance of \tori by adding $t-4$ new terminals and $k$ edges linking each of them with $r$. It hence suffices to prove that \fourori is NP-hard.  We prove this by a reduction from \mthreeori , which is NP-hard by \Cref{etxdrcdzfvzugbu}. Let $(G,S,r,k,Y)$ be an instance of \mthreeori . We now create an instance $(G',S',r,k')$ of \fourori. We first let $V(G')$ contain $V(G)$. Next, for every $e=yy'\in \delta_G(Y)$ with $y \in Y$, we let $V(G')$ contain two more vertices $a_e$ and $b_e$. We further let $V(G')$ contain 4 more vertices $s_1',s'_2,s_3'$ and $s^*$. We now define two edge sets $E_1$ and $E_2$. First, we let $E_1$ contain all edges in $E(G)\setminus \delta_G(Y)$.
    Next, we let $E_1$ contain an edge connecting $y$ and $a_e$, an edge connecting $a_e$ and $b_e$, and an edge connecting $b_e$ and $y'$ for all $e=yy'\in \delta_G(Y)$ with $y \in Y$. Further, we let $E_1$ contain an edge connecting $b_e$ and $s^*$ for all $e \in \delta_G(X)$ and $k$ edges connecting $r$ and $s^*$. Finally, we let $E_1$ contain a set of $k$ parallel edges connecting $s_i$ and $s_i'$ for $i \in [3]$. We now let $E_2$ consist of an edge connecting $r$ and $a_e$ and an edge connecting $a_e$ and $s'_i$ for all $e \in \delta_G(Y)$ and all $ i \in [3]$. We set $E(G')=E_1 \cup E_2, S'=\{s'_1,s'_2,s'_3,s^*\}$ and $k'=k+d_G(Y)$. This finishes the description of $(G',S',r,k')$; see 
    Figure~\ref{fig:t5_hardness} for an illustration.
    For $i \in [2]$, we use $G'_i$ for the graph defined by $V(G'_i)=V(G')$ and $E(G'_i)=E_i$. It is easy to see that $(G',S',r,k')$ can be constructed from $(G,S,r,k,Y)$ in polynomial time and $|S'|=4$ holds. We claim that $(G',S',r,k')$ is a yes-instance of \fourori if and only if $(G,S,r,k,Y)$ is a yes-instance of \mthreeori.

    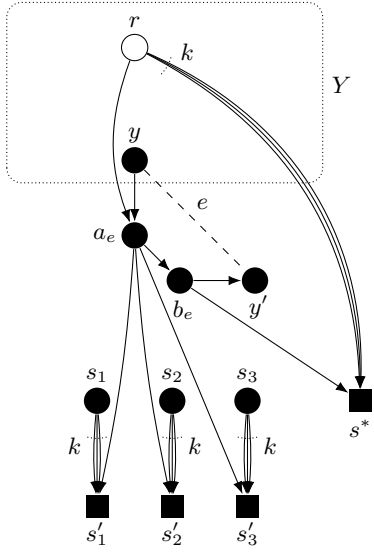
\begin{figure}
        \centering
        \begin{tikzpicture}
		\tikzset{
			every circle node/.style={minimum size=0.35cm, inner sep=0cm, font=\small},
			every rectangle node/.style={minimum size=0.3cm, inner sep=0cm, font=\small},
			>={Latex[scale=1]}
		}
		
		\node[circle, draw, minimum size=1cm, inner sep=1cm, label={$r$}] (r) at (2.5,0) {};
		
		\def\angles{-5, 0, 5} 
		\foreach \i in {1,2,3} {
			\node[circle, fill, label={90:$s_\i$}] (s\i) at (\i+1,-4.7) {};
			\node[fill, label={[yshift=-0.05cm]270:$s'_\i$}] (sp\i) at ($(s\i)+(0, -1.4)$) {};
			\foreach \angle in \angles {
				\draw[->] (s\i) to[bend right=\angle] (sp\i);
			}
			\draw[densely dotted] ($(s\i) + (255:0.5)$) 
			arc (255:285:0.5);
		}
		\foreach \i in {1} {
			\node at ($(s\i)+(-0.3,-0.6)$) {\small $k$};
		}
		\foreach \i in {2,3} {
			\node at ($(s\i)+(0.3,-0.6)$) {\small $k$};
		}
		%	\node at ($(s1)!0.5!(sp4)$) {$\ldots$};
		
		\def\angles{0,2,4} 
		\node[fill, label={[yshift=-0.05cm]270:$s^*$}] (sstar) at (5.5,-4.7) {};
		\foreach \angle in \angles {
			\draw[->] (r) to[bend left=\angle+30] (sstar);
		}
		\draw[densely dotted] ($(r) + (315:0.5)$) 
		arc (315:345:0.5);
		\node at ($(r)+(0.7,-0.05)$) {$k$};
		\draw[rounded corners=2mm, densely dotted, draw] (0.8,-1.8) rectangle (5,0.6);
		\node at (5.25, -0.5) {$Y$};
		
		\node[circle, fill, label={[]90:$y$}] (y1) at (2.5,-1.5) {};
		\node[circle, fill, label={[xshift=-0.05cm]180:$a_e$}] (a1) at ($(y1)+(0, -1)$) {};
		\node[circle, fill, label={[xshift=0.05cm, yshift=-0.05cm]270:$b_e$}] (b1) at ($(a1)+(0.6, -0.6)$) {};
		\node[circle, fill, label={[xshift=0.05cm]270:$y'$}] (z1) at ($(y1)+(1.6,-1.6)$) {};
		\draw[->] (y1) to (a1);
		\draw[->] (a1) to (b1);
		\draw[->] (b1) to (z1);
		\draw[dashed] (y1) -- (z1);	
		\node[label={[xshift=0.1cm, yshift=-0.1cm]$e$}] at ($(y1)!0.5!(z1)$) {};
		
		\draw[->, bend right=20] (r) to (a1);
		\draw[->] (b1) to (sstar);
		\draw[->, bend left=4] (a1) to (sp1);
		\draw[->, bend right=4] (a1) to (sp2);
		\draw[->, bend right=0] (a1) to (sp3);
	\end{tikzpicture}
        \caption{An illustration of the proof of \cref{tfix}.} \label{fig:t5_hardness}
    \end{figure}

    First suppose that $(G,S,r,k,Y)$ is a yes-instance of \mthreeori, so there exists an orientation $\dG$ of $G$ such that all edges in $\delta_G(Y)$ are oriented away from $Y$ and $\lambda_{\dG}(r,s_i)\geq k$ holds for all $i \in [3]$. We now define an orientation $\dGp$ of $G'$ as follows: first, we let every $e \in E(G)\setminus \delta_G(Y)$ have the same orientation it has in $\dG$. Next, we orient all the edges incident to $r$ away from $r$ and we orient all edges in $\delta_{G'}(S')$ toward $S'$. Finally, for every $e=yy' \in \delta_G(Y)$ with $y \in Y$, we orient the edge connecting $y$ and $a_e$ from $y$ to $a_e$, we orient the edge connecting $a_e$ and $b_e$ from $a_e$ to $b_e$ and we orient the edge connecting $b_e$ and $y'$ from $b_e$ to $y'$. This finishes the description of $\dGp$.
    We let $\dGpo$ and $\dGpt$ denote the inherited orientations of $G_1'$ and $G_2'$, respectively. We still need to show that $\lambda_{\dGp}(r,s')\geq k'$ holds for every $s' \in S'$.
    We clearly have $\lambda_{\dGp}(r,s^*) \ge k'=k+d_G(Y)$ since $\dGp$ contains $k$ parallel $rs^*$ arcs, and $ra_eb_es^*$ for $e \in \delta_G(Y)$ is a collection of $d_G(Y)$ arc-disjoint directed $r$-$s^*$ paths in $\dGp$. Now we consider some $i \in [3]$ and show that $\lambda_{\dGp}(r, s'_i) \ge k'$.
    As $\lambda_{\dG}(r, s_i)\ge k$, there exist arc-disjoint directed $r$-$s_i$ paths $P_1, \ldots, P_k$ in $\dG$.
    For each $j \in [k]$, $P_j$ contains a unique arc $e = yy'\in \delta_G(Y)$ with $y \in Y$. 
    Let $P'_j$ denote the path obtained from $P_j$ by replacing the arc $yy'$ with the path $ya_eb_ey'$.
    Then, $P'_1,\ldots, P'_k$ is a collection of $k$ arc-disjoint directed $r$-$s_i$ paths in $\dG'_1$.
    As $\dGpo$ contains $k$ parallel $s_is'_i$ arcs, we have $\lambda_{\dG'_1}(r, s'_i)\ge k$.
    We also have $\lambda_{\dGpt}(r, s'_i)\ge d_G(Y)$ as $ra_es'_i$ is a directed $r$-$s'_i$ path in $\dGpt$ for each $e \in \delta_G(Y)$.
    It follows that $\lambda_{\dGp}(r, s'_i) \ge \lambda_{\dGpo}(r, s_i')+\lambda_{\dGpt}(r,s_i')\ge k+d_G(Y)=k'$.
    Hence $(G',S',r,k')$ is a yes-instance of \fourori.

    Now suppose that $(G',S',r,k')$ is a yes-instance of \fourori, so there exists an orientation $\dGp$ of $G'$ with $\lambda_{\dGp}(r,s')\geq k'$ for all $s' \in S'$. We first collect some structural properties of $\dGp$. First observe that for every $s'\in S'$, we have $k'\leq \lambda_{\dGp}(r,s')\leq d_{\dGp}^{in}(s')\leq d_{G'}(s')=k'$. Hence equality holds throughout and all edges in $\delta_{G'}(S')$ are oriented toward $S'$ in $\dGp$. We now prove two claims.
    
    \begin{cla}\label{drztfugziuh}    
        For all $e\in \delta_G(Y)$, the edge connecting $a_e$ and $b_e$ is oriented from $a_e$ to $b_e$ in $\dGp$.
    \end{cla}
    \begin{claimproof}
        Let $Z=Y \cup \{a_e:e\in \delta_G(Y)\}\cup \{s'_1,s'_2, s'_3\}$. As all edges incident to $s'_i$ are oriented toward $s'_i$ for $i \in [3]$ and by construction, we obtain that $\delta_{\dGp}^{out}(Z)$ consists of the orientations of the $k$ edges connecting $r$ and $s^*$ and orientations of edges of the form $a_eb_e$ for some $e \in \delta_G(Y)$. As $r \in Z, s^*\in V(G')\setminus Z$ and $\lambda_{\dGp}(r,s^*)\geq k'$ by assumption, we obtain $k'\leq \lambda_{\dGp}(r,s^*)\leq d_{\dGp}^{out}(Z)\leq d_G(Y)+k=k'$. Hence equality holds throughout and the statement follows.
    \end{claimproof}

    In the following, let $\dGpo$ and $\dGpt$ be the orientations of $G'_1$ and $G'_2$ inherited from $\dGp$, respectively. 
   
    \begin{cla}\label{dtfuzgu}
        For every $i \in [3]$, we have $\lambda_{\dGpo-S'}(r,s_i)\geq k$.
    \end{cla}
    \begin{claimproof}
        For every $e \in \delta_G(Y)$, we have that the edge connecting $a_e$ and $s_i'$ is oriented from $a_e$ to $s_i'$ and clearly, we may suppose that the edge connecting $r$ and $a_e$ is oriented from $r$ to $a_e$. Hence, repeatedly applying \cref{vuzuighj}, we obtain that $\lambda_{\dGpo}(r,s_i')=k'-d_G(Y) =k$. As all edges in $\delta_{G'}(S')$ are oriented toward $S'$ in $\dGp$ and $S'$ is an independent set in $G'$ by construction, we obtain that $\lambda_{\dGpo-(S'\setminus s_i)}(r,s_i')\geq k$. Now let $Z \subseteq V(\dGpo-S')$ with $r \in Z$ and $s_i \in V(\dGpo-S')\setminus Z$. As $\lambda_{\dGpo-(S'\setminus s_i')}(r,s_i')\geq k$, we obtain that $d_{\dGpo-(S'\setminus s_i')}^{out}(r,s_i')(Z)\geq k$. As by construction all edges in $E(G_1'-(S'\setminus s_i'))$ that are incident to $s'_i$ are also incident to $s_i$, we obtain $d_{\dGpo-S'}^{out}(Z)\geq d_{\dGpo-(S'\setminus s_i')}^{out}(Z)\geq k$. Hence the statement follows.
    \end{claimproof}

We are now ready to finish the proof. Let $\dG$ be the orientation of $G$ in which every edge in $E(G)\setminus\delta_G(Y)$ has the same orientation as in $\dGp$ and all edges in $\delta_G(Y)$ are oriented away from $Y$. It suffices to prove that $\lambda_{\dG}(r,s_i)\geq k$ holds for all $i \in [3]$. Now let $i \in [3]$ and let $Z \subseteq V(G)$ be a set with $r \in Z$ and $s_i \in V(G)\setminus Z$. Further, let $Z'$ be the set that contains $Z$ and that for every $e=yy'\in \delta_G(Y)$ with $y \in Y$ contains $a_e$ if $Z$ contains $y$ and contains $b_e$ if $Z$ contains $y'$. By Claim \ref{dtfuzgu}, we have $d_{\dGpo-S'}(Z')\geq k$. Now consider some $a \in \delta_{\dGpo-S'}(Z')$. It follows by construction that either $a$ is an orientation of an edge in $E(G)\setminus \delta_G(Y)$ or $a$ is an orientation of the edge connecting $a_e$ and $b_e$ for some $e=yy' \in \delta_G(Y)$ with $y \in Y$. In the first case, we obtain that $a$ is also contained in $\delta_{\dG}^{out}(Z)$. In the latter case, by Claim \ref{drztfugziuh}, we obtain that $a_e \in Z'$ and $b_e \in V(G')\setminus Z$. It hence follows by construction that $yy' \in \delta_{\dG}^{out}(Z)$. We obtain that $d_{\dG}^{out}(Z)\geq d_{\dGp}(Z')\geq k$. We obtain that $\lambda_{\dG}(r,s_i)\geq k$ and hence $(G,S,r,k,Y)$ is a yes-instance of \mthreeori. 
\end{proof}

%%%%%%%%%%%%%%%%%%%%%%%%%%%%%%%%
\section{Conclusion}
\label{sec:conclusion}
%%%%%%%%%%%%%%%%%%%%%%%%%%%%%%%%

In this paper, we studied the Steiner Rooted Orientation problem, motivated by ensuring reliable connectivity from a root to a set of terminals. We developed an algorithm running in time $f(k,t)\cdot n^{O(1)}$, establishing fixed-parameter tractability with respect to the required connectivity $k$ and the number of terminals $t$. On the other hand, we showed that the problem becomes NP-hard as soon as either $k$ or $t$ is part of the input, which implies that our algorithm is optimal from a parameterized point of view. Our results also extend to the more general setting where a local requirement $R(u,v)$ is given for every ordered pair $(u,v)$ of vertices, and the parameter is the total demand $\alpha=\sum_{(u,v)\in V\times V}R(u,v)$. We close the paper by mentioning some open problems.

\begin{enumerate}\itemsep0em
    \item While our results provide an almost complete picture for edge-disjoint rooted connectivity, some cases remain unresolved. The most immediate question concerns \Cref{tfix}: What is the complexity of \kori when $k$ is part of the input and $t\in\{2,3\}$? 
    \item Another direction is to consider orientations that guarantee internally vertex-disjoint paths rather than arc-disjoint ones, or to allow some edges to be pre-oriented. These variants are at least as hard as the arc-disjoint setting considered here, raising the following questions: Does an analogue of our fixed-parameter tractability result hold when we require $k$ internally vertex-disjoint paths from the root to each terminal, or when dealing with mixed graphs where unoriented and pre-oriented edges coexist? Do the corresponding rooted orientation problems remain tractable when parameterized by $k$ and $t$?
\end{enumerate}

\medskip

%%%%%%%%%%%%%%%%%%%%%%%%%%%%%%%%
\paragraph{Acknowledgment.} 
%%%%%%%%%%%%%%%%%%%%%%%%%%%%%%%%
András Imolay was supported by the EKÖP-KDP-25 University Research Scholarship Program, Cooperative Doctoral Program of the Ministry for Culture and Innovation, from the source of the National Research, Development and Innovation Fund. The research received further support from the Lend\"ulet Programme of the Hungarian Academy of Sciences -- grant number LP2021-1/2021, from the Ministry of Innovation and Technology of Hungary from the National Research, Development and Innovation Fund -- grant numbers ADVANCED 150556 and ELTE TKP 2021-NKTA-62, and from the Dynasnet European Research Council Synergy project -- grant number ERC-2018-SYG 810115. This work was supported in part by EPSRC grant EP/X030989/1.

\paragraph{Data Availability.} No data are associated with this article. Data sharing is not applicable to this article. %Needed for the EPSRC grant.

%%%%%%%%%%%%%%%%
\bibliographystyle{abbrv}
\bibliography{steiner}
%%%%%%%%%%%%%%%%

\end{document}